%% file: paper.tex
\documentclass[12pt]{article}
\usepackage{amsmath}
\usepackage{amsthm}
\usepackage{amsfonts}
\usepackage{amssymb}
\usepackage{graphicx,psfrag,epsf}
\usepackage{subcaption} % subfigures and subcaptions
\usepackage{enumerate}
\usepackage{natbib}
\usepackage{url} % not crucial - just used below for the URL 
\usepackage{tikz}
\usepackage{algpseudocode}
\usepackage{algorithmicx} % algorithm and algorithmic environment
\usepackage{algorithm}    % \listofalgorithms command 
\usepackage[colorlinks=true]{hyperref}
\usepackage{booktabs}
\usepackage{multirow}
\usepackage{rotating}
\usepackage{bm}             % bold symbols, command \bm
\usepackage{dsfont} 		% for indicator
\usepackage{color, graphics}
\usepackage{enumitem}       % customizable enumerate labels
\usepackage{comment}

\newcommand{\blind}{0}

\newcommand{\paul}[2]{\textcolor{red}{#1} \iffalse #2 \fi}
\input{macros}

\begin{document}

\def\spacingset#1{\renewcommand{\baselinestretch}%
{#1}\small\normalsize} \spacingset{1}

%%%%%%%%%%%%%%%%%%%%%%%%%%%%%%%%%%%%%%%%%%%%%%%%%%%%%%%%%%%%%%%%%%%%%%%%%%%%%%

\if0\blind
{
  %\title{\bf Textual Positioning: The Mosaic of Influences}
\title{A Structural Text-Based Scaling Model for Analyzing Political Discourse} 
\author{
  Jan Vávra\textsuperscript{1,2}  \hspace{0.5em} Bernd Prostmaier\textsuperscript{1} \hspace{0.5em} Bettina Grün\textsuperscript{2} \hspace{0.5em} Paul Hofmarcher\textsuperscript{1} 
  \vspace{2em}
  \date{
  \textsuperscript{1}Department of Economics,\\
  Paris Lodron University Salzburg\\\vspace{0.5em}
  \textsuperscript{2}Institute for Statistics and Mathematics,\\
  WU Vienna University of Economics and Business\\\vspace{0.5em}
  % \textsuperscript{3}\emph{BMW Group AG}\\[3ex]
  \today
  }
}
\maketitle
} \fi
\if1\blind
{
\title{A Structural Text-Based Scaling Model for Analyzing Political Discourse} 
\author{
  \vspace{2em}
  \date{
  \today
  }
}
\maketitle
} \fi

%\bigskip
\begin{abstract}
%% 200 words
Scaling political actors based on their individual characteristics and behavior helps profiling and 
grouping them as well as understanding changes in the political landscape. In this paper we introduce the Structural Text-Based Scaling (STBS) model to infer ideological positions of speakers for latent topics from text data. We expand the usual Poisson factorization specification for topic modeling of text data and use flexible shrinkage priors to induce sparsity and enhance interpretability. We also incorporate speaker-specific covariates to assess their association with ideological positions. Applying STBS to U.S.~Senate speeches from Congress session 114, we identify immigration and gun violence as the most polarizing topics between the two major parties in Congress. Additionally, we find that, in discussions about abortion, the gender of the speaker significantly influences their position, with female speakers focusing more on women's health. We also see that a speaker's region of origin influences their ideological position more than their religious affiliation.
% from session 114 allows us to infer key drivers of polarization and to explore how speaker-specific characteristics influence their ideological positions. The STBS model thus contributes to a more comprehensive understanding of political discourse.
\end{abstract}

\noindent%
{\it Keywords:}  ideal point, Poisson factorization, political discourse, topic model, variational inference, partisanship
% 3 to 6 keywords, that do not appear in the title
\vfill

\newpage
%\spacingset{1.45} % DON'T change the spacing!
%\spacingset{1.9} % DON'T change the spacing!
\spacingset{1.0} % DON'T change the spacing!
\input{sec01_Intro_Paul}

\input{sec02_STBSM}
\input{sec03_inference}
\input{sec04_senate_Paul}

\input{sec05_conclusion}

\if0\blind
{
\section*{Acknowledgments}
This research was supported by the Jubil\"aumsfonds of the Oesterreichische Nationalbank (OeNB, grant no.~18718).
}
\fi
\if1\blind
{

}
\fi

\begin{comment}
\clearpage % to be commented
\bigskip
\begin{center}
{\large\bf SUPPLEMENTARY MATERIAL}
\end{center}

\begin{description}

\item[Title:] Brief description. (file type)

\item[R-package for  MYNEW routine:] R-package ÒMYNEWÓ containing code to perform the diagnostic methods described in the article. The package also contains all datasets used as examples in the article. (GNU zipped tar file)

\item[HIV data set:] Data set used in the illustration of MYNEW method in Section~ 3.2. (.txt file)

\end{description}
\end{comment}

% \clearpage % to be commented

\bibliographystyle{agsm}
\bibliography{literature}

\include{paper.bbl}
\newpage
\appendix
\input{appendix}
\end{document}

%% file: macros.tex
\theoremstyle{plain}

\newtheorem{athm}{Theorem}
\newtheorem{alemma}[athm]{Lemma}

\definecolor{red2}{rgb}{0.9333333, 0, 0}
\definecolor{mediumblue}{rgb}{0, 0, 0.8039216}
\definecolor{magenta4}{rgb}{0.545098, 0, 0.545098}
\newcommand{\R}{\mathbb{R}}

\DeclareMathOperator{\E}{\mathsf{E}}

\DeclareMathOperator{\diag}{\mathsf{diag}}

\newcommand{\T}[1]{#1^\top}        

\newcommand{\norm}[2]{\mathsf{N}\left(#1, \,#2 \right)} %normal 1=E,2=var
\newcommand{\knorm}[3]{\mathsf{N}_{#1}\left(#2, \,#3 \right)} %MV normal 1=dim k, 2=E, 3=var
\newcommand{\gammadist}[2]{\Gamma\left(#1, #2 \right)} %Gamma distribution 1=shape alpha, 2=rate beta
\newcommand{\Pois}[1]{\mathsf{Pois}\left(#1 \right)} %Poisson 1=lambda
\newcommand{\ELBO}[1]{\mathsf{ELBO}\left(#1\right)} %ELBO function of parameter 1
\newcommand{\Eq}{\E_{q_{\bm\phi}}}

\newcommand{\shp}{\mathrm{shp}}
\newcommand{\rte}{\mathrm{rte}}
\newcommand{\loc}{\mathrm{loc}}
\newcommand{\scl}{\mathrm{scl}}
\newcommand{\sclsq}{\mathrm{var}}
\newcommand{\covm}{\mathrm{cov}}
\DeclareMathAlphabet{\mathbbb}{U}{bbold}{m}{n}
\newcommand{\ideal}{\mathbbb{i}}

%% file: sec01_Intro_Paul.tex
\section{Introduction}
% In this article we present the structural text based scaling model (STSM) which allows to detect ideological positions of lawmakers on topical level but also allows to incorporate speaker-specific covariates which might influence their ideological positions. 

Estimating ideological positions of lawmakers has a long tradition in political science. 
\citet{Poole:1985} proposed a ``scaling procedure'' to estimate ideological positions of lawmakers based on their voting behavior. Dynamic weighted nominal three-step estimation \citep{McCarty+Poole+Rosenthal:1997}, an extension of this procedure, results in the DW-Nominate scores that are widely accepted as benchmark ideological positions both on party level as well as on individual level \citep[see, e.g.,][]{Nominatescore, vox, vox2}.
Legislative votes, however, provide limited information on the latent ideological positions because voting behavior on individual level is often not documented and lawmakers rarely diverge from party-line voting due to robust party discipline \citep{Hug_2010}. Consequently, roll-call analysis for inferring the ideological positions adopted by legislators both within and across parties is of limited value \citep[see, e.g.,][]{wordshoal_2016}.%%von wordshoal paper

Text-based scaling models are a promising alternative method to discern ideological stances based on political discussions. 
The \emph{wordfish} model \citep{Slapin_2008} 
is specifically tailored to infer the political positions of authors from their texts but requires political texts about a single issue/topic as input.
\cite{wordshoal_2016} extend this approach to multiple issues/topics by proposing the \emph{wordshoal} model which allows to analyze a collection of labeled texts. The \emph{wordshoal} model fits separate models to each group of texts with the same label and combines the fitted models in a one-dimensional factor analysis to obtain ideological positions. 

\cite{Vafa_etal_2020} propose the text-based ideal point model (TBIP), which constitutes an extension of \cite{wordshoal_2016} and which merges the idea of topic models \citep[see, e.g.,][]{blei_etal_2003} and text-based scaling models. TBIP does not require single issues/topics to be assigned to the texts which makes this method applicable to corpora where no topic labeling of the speeches exists. The model simultaneously learns both, ideological positions (referred to as ideal points or scaling positions) for the authors and latent topics of the texts and thus allows inferring how the per-topic word choice changes as a function of the political position of the author. Hereby the model relies on the assumption that the latent ideological position of an author is the same across all topics. This assumption of one identical ideological position across all topics is, however, very restrictive, as it does not allow a speaker to be more or less polarizing across topics. % Also \citet{Taddy_etal_2019} find that polarization varies across topics in political speech using a supervised language model.

%\cite{Hofmarcher+Adhikari+Gruen:2022} apply the TBIP model by aggregating the single scaling positions to a measure of partisanship between parties. In line with 
%\citet[][]{Taddy_etal_2019}, they find that partisanship increased in recent years. 
% STM is designed to uncover the thematic structure of a corpus of texts, identifying latent topics and how they are associated with document-level metadata, such as political party affiliation or other covariates.

Further, focus is usually not only on inferring the ideological positions of individual authors, but to perform a structural analysis and assess how specific author characteristics impact on the ideological positions, e.g., how gender, religion and political affiliation may influence the ideal point of a speaker. In the case of roll-call votes, e.g., \cite{Swers:1998, Welch:1985} study the effect of gender on political ideology, and \cite{Fastnow} discuss the influence of religion on political behavior using DW-Nominate scores.

This paper proposes the structural text-based scaling (STBS) model which addresses both mentioned caveats: First, the model estimates topic-specific ideological positions of political actors, i.e., the scaling position for one topic may differ from that of another topic for a given author. Second, the model incorporates author-specific covariates to characterize the topic-specific ideological positions.
Because of this second extension, we refer to this model as the \emph{structural} text-based scaling (STBS) model in reference to \cite{Roberts+Stewart+Tingley:2013, Roberts_etal_JASA:2016} who introduced the structural topic model (STM) which allows to include covariates to characterize the word distributions and the topic distributions. 
Including these extensions results in a model which enables a more nuanced and holistic understanding of how various factors shape and influence ideological positions within the realm of political discourse, ultimately contributing to a richer and more comprehensive analysis of political language. 

Our model builds on Poisson factorization topic models \citep{Gopalan:2014} to introduce topic-specific scaling positions as well as speaker-specific covariates to characterize these positions.
Classical topic models \citep{blei_etal_2003, Roberts_etal_JASA:2016} estimate per document topic proportions $\mathbf{\theta}$ and per topic term distributions $\mathbf{\beta}$. In a Poisson factorization topic model the products thereof serve as Poisson rates for the observed counts of terms on document-term level. We will refer to this product as the neutral Poisson rates. 
To account for ideological positions on topic level, we introduce scaling factors to the neutral Poisson rates which result from the speaker-specific topical scaling and the polarity rates of the vocabulary. 
These scaling factors capture the idea of framing, i.e., that the communicator will emphasize certain aspects through word usage and that specific term usage when discussing a topic conveys political messages \citep[see, e.g.,][]{Entman}. % This implies that an author's word choice for a particular topic is affected by their ideological position.

We carefully select the hierarchical prior structure for the latent variables in our model to achieve modeling aims and ensure computational efficiency. Specifically, for topic prevalence, we make use of gamma distributions with author-specific rates to account for differences in speech length. The topical content $\mathbf{\beta}$ is modeled through a hierarchical structure to allow for a very flexible topic-term distribution. The neutral Poisson rates are scaled through term-specific polarity rates and author-specific scaling rates and we impose a hierarchical shrinkage prior on the polarity rates and on the scaling rates by imposing the triple gamma prior \citep{CadonnaFruehwirth-SchnatterKnaus_triple_gamma:2019} on regression coefficients corresponding. In particular, this prior specification aligns with the modeling aim that only a few terms are crucial to determine the polarity of a topic. 

As with other topic models or scaling models, the exact posterior for the proposed model is intractable. We make use of variational inference \cite[VI;][]{BleiVI:2017} to approximate the posterior distribution of the model parameters. 
% VI is a very flexible and fast method to approximate the posterior distribution in cases where the set of parameters is large, and VI is widely used in text analysis \citep{BleiVI:2017, Vafa_etal_2020}. 
In our proposed VI algorithm, we update as many parameters as possible using coordinate ascent variational inference \citep{gopalanhofmanblei2015scalableHPoisF} and the remaining parameters using automatic differentiation gradient descent \citep{Vafa_etal_2020}. To facilitate the use of the model, we provide a GitHub repository \citep{Github_STBS} which implements model fitting, summarization and visualization.

We demonstrate our model using U.S.\ Senate speeches from the 114th Congress session. This dataset has been made available by \cite{Gentzkow+Shapiro+Taddy:2018} and used in previous studies to investigate in particular the polarization or partisanship between the two major parties as manifested in different word usage \citep{Gentzkow_etal_2019, Stewart_2023, Vafa_etal_2020}. Using our model, we are interested in characterizing which topics of political discourse are the main drivers of polarization and we identify immigration and gun violence to be the most polarizing topics. We also evaluate the effect of how speaker-specific covariates like gender, experience, region of origin or religious affiliation influence the scaling position of the latent topics. For example, we find that the gender of the speaker significantly influences their position when discussing the topic about abortion.
%Our model identifies the topics about immigration and gun violence as the most polarizing ones between the two major parties in Congress. Additionally, we find that the gender of the speaker significantly influences their position when discussing the topic about abortion, where female speakers often focus on women's health, while male Republicans tend to frame the conversation around the unborn child to foster emotional feelings. We also find that the speaker's region of origin has a greater influence on their ideology than their religious affiliation.

The paper is structured as follows: In Section~\ref{sec:STBSM}, we present the STBS model including the data generative part and the hierarchical prior specification. We provide a variational inference algorithm to approximate the posterior distribution of the model parameters in Section~\ref{sec:inference}. Section~\ref{sec:senate} provides the empirical results when applying our model to the speeches given in the 114th session of the U.S.~Senate. We present our results and compare them to the TBIP model \citep{Vafa_etal_2020}. We conclude in Section~\ref{sec:conclusion}.

%% file: sec02_STBSM.tex
\section{Structural text-based scaling (STBS) model} \label{sec:STBSM}

% Political science has commonly employed scaling models, such as \emph{wordfish} and \emph{wordshoal}, to estimate ideological positions on a one-dimensional scale based on text data. Statistics and computer science make use of topic models, like Latent Dirichlet Allocation \citep[LDA;][]{blei_etal_2003} and extensions thereof, to analyze text data by clustering words into topics and representing documents as mixtures of topics. The integration of both approaches  was put forward by  \cite{Vafa_etal_2020} with the TBIP model which bridges scaling methods with topic modeling to identify ideological positions and hidden topical content in texts. 

% The Structural Text-Based Scaling (STBS) model contributes to this literature by two extensions: First, building on the ideas of  structural topic models \citep{Roberts_etal_JASA:2016}, our model includes author-specific covariates to characterize the ideological positions based on a regression model. Second, the STBS model relaxes the assumption that the ideological position is the same for a specific author across all latent topics and allows for topic-specific ideological positions of each author. We combine this flexible model formulation with hierarchically well-structured prior distributions to regularize the model specification and include prior model assumptions.

We first specify the data structure required for STBS. Our data input consists of a~corpus of $D$ documents (speeches) together with information on the author (speaker) of each document.
Based on the bag-of-words assumption, the corpus is represented by a document-term matrix $\mathbb{Y}$ of shape $D \times V$, where $V$ is the number of unique terms (e.g., bigrams) in the corpus. Each entry of the document-term matrix represents the frequency count $y_{dv}$, i.e., the number of appearances of term $v \in \{1, \ldots, V\}$ in document $d \in \{1, \ldots, D\}$. 
The document-term matrix usually is a sparse matrix; the most frequent element is zero given that both $D$ and $V$ are usually in the thousands and documents are rather short. There are $A$ authors of the documents, and $a_d \in \{1, \ldots, A\}$ indicates the author of document $d$. For each author a vector of author-specific covariates of length $L$, $\bm x_a$ is available. These covariates are combined to the covariate matrix $\mathbb{X}$.

The elements $y_{dv}$ of $\mathbb{Y}$ are assumed to follow independent Poisson distributions, i.e., 
\begin{equation}\label{eq:Pois_counts}
    y_{dv} \sim \Pois{\lambda_{dv}} 
    \qquad \text{where} \qquad 
    \lambda_{dv} = \sum\limits_{k=1}^K \lambda_{dkv}
    \qquad \text{and} \qquad
    \lambda_{dkv} = \theta_{dk} \beta_{kv} \exp\{\eta_{kv} \ideal_{a_d k}\},
\end{equation}
with $K$ the number of latent topics.
The positive per-document topic intensities $\theta_{dk}$ reflect the contribution of latent topic~$k$ to the frequency counts of document~$d$. The term intensities of each topic for author $a_d$ are given by $\beta_{kv} \exp\{\eta_{kv} \ideal_{a_d k}\}$. Each author~$a$ has ideological positions $\ideal_{ak}$ for $k=1,\ldots,K$, which are specific to each topic. 

For an author~$a$ with latent ideological position $\ideal_{ak} = 0$ towards topic~$k$, the term intensities are captured solely by the intensities $\beta_{kv}$. Hence, we refer to them as \emph{neutral} term intensities. An ideologically shifted author~$a$ with $\ideal_{ak} \neq 0$ towards topic~$k$ has the neutral term intensities multiplied by the \emph{ideological factor} $\exp \{ \eta_{kv} \ideal_{ak} \}$, where $\eta_{kv}$ are the corresponding term-specific \emph{polarity values} for a unit change in ideological position $\ideal_{ak}$.\footnote{Note the possible reparametrization $\eta_{kv} \ideal_{ak} = (- \eta_{kv}) (- \ideal_{ak})$ that brings identifiability issues. Hence, prior to the analysis the meaning of having a positive position has to be fixed.}
For a given topic $k$, an author with an ideological position $\ideal_{ak}>0$ will use terms more often than a neutral speaker, if those terms have a polarity score $\eta_{kv}>0$. If the ideological position and the polarity score have opposite signs, the frequency of these terms is decreased compared to a neutral speaker. 
To give an intuitive example, a topic $k$ about gun violence may have \emph{gun violence} or \emph{mass shooting}  as frequent neutral terms (i.e., high $\beta_{kv}$ values). A person who aligns ideologically with liberal gun laws will also include terms such as \emph{individual freedom} and \emph{amendment rights} when talking about this topic. By contrast, a person who wants more restrictive gun laws will include terms like \emph{background checks} or \emph{watch list}. The term $\exp\{\eta_{kv} \ideal_{a_d k}\}$ captures those nuances. 

For the ideological positions $\ideal_{ak}$ on topic level, we assume that they can be characterized by author-specific covariates, i.e.,
 \begin{equation}
     \label{eq:Pois_counts2}
    \ideal_{ak} \sim \norm{\T{\bm x_a} \bm \iota_k}{\left(I_a^2\right)^{-1}},
 \end{equation}
where $\bm x_a$ is the covariate vector of author~$a$, $\bm \iota_k$ are the regression coefficients specific to topic~$k$ and $I_a^2$ are the author-specific precisions which capture how close the author-specific ideological positions are to the predicted ideological positions given the covariates. Having topic-specific regression coefficients $\bm \iota_k$ allows covariate effects to differ across topics. For example, gender might not impact word choice for most topics but when it comes to discussing women rights clear differences in word usage are induced by gender. This feature thus allows us to differentiate across topics how covariates impact the ideological positions.

% \textcolor{red}{A paragraph explaining the relationship between $\eta$ and $\ideal$. Supported by some specific example. }

In STBS, the data generative Equations~\eqref{eq:Pois_counts} and \eqref{eq:Pois_counts2} are complemented by a careful choice of hierarchical priors on the parameters. This hierarchical prior structure increases the flexibility of the model, helps identify the parameters of interest and -- through the inclusion of structural parameters -- eases interpretation. 

Documents and authors may differ in their verbosity. To prevent that the ideological positions capture differences in verbosity between authors, \citet{Vafa_etal_2020} multiplied the Poisson rate $\lambda_{dv}$ with a fixed \emph{verbosity} term specific to each author. %\textbf{Whats the problem with VAFA approach - maybe one sentence why we need a more flexible one} 
We propose a~more flexible approach inspired by scalable Poisson factorization \citep{gopalanhofmanblei2015scalableHPoisF}, where a~hierarchical prior for the intensities is used. 
\citet{gopalanhofmanblei2015scalableHPoisF} specified the document intensities to be a-priori gamma-distributed with fixed hyperparameters, e.g., $\theta_{dk} \sim \gammadist{0.3}{0.3}$. We extend this to account for differences in document length by using a specification with fixed shape $a_\theta$ and author-specific rates $b_{\theta a}$, i.e., $\theta_{dk} \sim \gammadist{a_\theta}{b_{\theta a_d}}$. This allows for different verbosity between authors. The rates $b_{\theta a}$ are additional model parameters with prior $b_{\theta a} \sim \gammadist{a_\theta'}{\frac{a_\theta'}{b_\theta'}}$, where both hyperparameters are fixed.

%$b_{\theta d}$, i.e., $\theta_{dk} \sim \gammadist{a_\theta}{b_{\theta d}}$. The latent auxiliary rates $b_{\theta d}$ are additional model parameters with prior $b_{\theta d} \sim \gammadist{a_\theta'}{\frac{a_\theta'}{b_\theta'}}$ where both hyperparameters are fixed. 

We also impose a hierarchical prior structure on the neutral term intensities $\beta_{kv}$ to capture differences in the frequency of term usage. 
The following hierarchical structure is imposed on the neutral topic-term frequencies $\beta_{kv}$. 
The simple gamma prior with fixed hyperparameters is extended to a gamma prior with flexible $v$-rates, i.e., $\beta_{kv} \sim \gammadist{a_\beta}{b_{\beta v}}$ and $b_{\beta v} \sim \gammadist{a_\beta'}{\frac{a_\beta'}{b_\beta'}}$. 
If a term $v$ is very popular in the corpus, the hierarchical prior accounts for this by a lower rate $b_{\beta v}$.
%\textbf{should we add one sentence why we need this or cite a paper.}.

We assume that only a few terms are crucial to express ideology. Therefore, we use a~shrinkage prior for the polarity values $\eta_{kv}$. This shrinkage prior has a spike at zero which is achieved by adjusting the precision parameter $\rho_k^2$ for the normal distribution and imposing a~gamma prior. Additional flexibility is obtained by adding a further level to the hierarchical prior distribution and imposing a prior on the rate of the gamma prior:
\begin{equation}\label{eq:prior_eta}
    \eta_{kv} \sim \norm{0}{\left(\rho_k^2\right)^{-1}} 
	\quad \text{with} \quad 
	\rho_k^2 \sim \gammadist{a_\rho}{b_{\rho k }}
	\quad \text{and} \quad 
	b_{\rho k } \sim \gammadist{a_\rho'}{\frac{\kappa_\rho^2}{2} \frac{a_\rho'}{a_\rho}}.
\end{equation}
This implies a very flexible prior structure which is equivalent \citep[see][]{Knaus:2023dynamic} to the triple-gamma prior by \citet{CadonnaFruehwirth-SchnatterKnaus_triple_gamma:2019}:
\begin{equation}\label{eq:prior_triple_gamma}
    \eta_{kv}^2 \sim \gammadist{\frac{1}{2}}{\frac{1}{2\sigma_k^2}}
    \; \text{and} \;
    \sigma_k^2 \sim \gammadist{a_\rho'}{\frac{a_\rho' b_{\sigma k}}{2}}
    \; \text{and} \;
    b_{\sigma k} \sim \gammadist{a_{\rho}}{\frac{a_{\rho}}{\kappa_\rho^2}},
\end{equation}
where the prior structure is based on the variances $\sigma_k^2 = \left(\rho_k^2\right)^{-1}$ instead of the precisions. With $a_\rho'<1$, the marginal density has a spike at zero and $b_{\sigma k}$ controls the amount of shrinkage, see \citet{Bitto_Fruehwirth-Schnatter_double_gamma:2019}. The triple-gamma prior makes the shrinkage flexible for each topic by having $b_{\sigma k}$ vary with $k$. The expected value $\E b_{\sigma k} = \kappa_\rho^2$ is fixed. 

%% hab ich rausgegeben
% This flexible shrinkage structure allows us to infer whether a topic is relevant for polarization. A large value of $b_{\sigma k}$ shrinks the polarity values completely towards zero, i.e., no ideological corrections are needed for this topic. By contrast, a small value of $b_{\sigma k}$ induces a non-negligible probability of obtaining $\eta_{kv}$ values far from zero. 
% Thus, together with the precision of the ideal points $I_a^2$, $b_{\sigma k}$ captures the overall polarity of a single topic. 

% Finally, we discuss the author-topic specific ideological positions $\ideal_{ak}$. 
Finally, continuing with the regression framework for the ideological positions $\ideal_{ak}$ specified in Equation~\eqref{eq:Pois_counts2}, we assume a gamma prior on the author-specific precisions  $I_a^2$, i.e.,  
\begin{equation} \label{eq:prior_ideal}
%     \ideal_{ak} \sim \norm{\T{\bm x_a} \bm \iota_k}{\left(I_a^2\right)^{-1}}
%     \quad \text{and} \quad
     I_a^2 \sim \gammadist{a_I}{b_I}.  
 \end{equation}

To capture the overall effect of a covariate on the ideological positions, we use the following prior structure:
\begin{equation}\label{eq:prior_iota}
    \iota_{kl} \sim \norm{\iota_{\bullet l}}{\left(\omega_l^2\right)^{-1}} 
    \qquad \text{and} \qquad
    \iota_{\bullet l} \sim \norm{0}{1},
\end{equation}
where the centrality parameters $\iota_{\bullet l}$ express the aggregated effect of covariate $l \in \{1, \ldots, L\}$ and $\omega_l^2$ are the precisions which are specific to each covariate. The lower the precision, the more does the effect of this covariate differ across topics. Assuming that the majority of the topics share a very similar effect for a covariate and only few topics substantially differ, we impose again the triple-gamma prior~\eqref{eq:prior_eta} yielding:
\begin{equation}\label{eq:prior_omega}
    \omega_l^2 \sim \gammadist{a_\omega}{b_{\omega l}}
    \qquad \text{and} \qquad 
    b_{\omega l} \sim \gammadist{a_\omega'}{\frac{\kappa^2_\omega}{2} \frac{a_\omega'}{a_\omega}}.
\end{equation}
In this case the shrinkage of the coefficients is not towards zero, but towards the centrality parameters $\bm \iota_{\bullet l}$.

\input{diagram2}

An overview of the STBS model is given by the graphical model in Figure~\ref{fig:hierarchy_STBSM}. The figure indicates the observed data, the latent variables of interest as well as the additional parameters present due to the imposed hierarchical prior structure. 

%% file: diagram2.tex
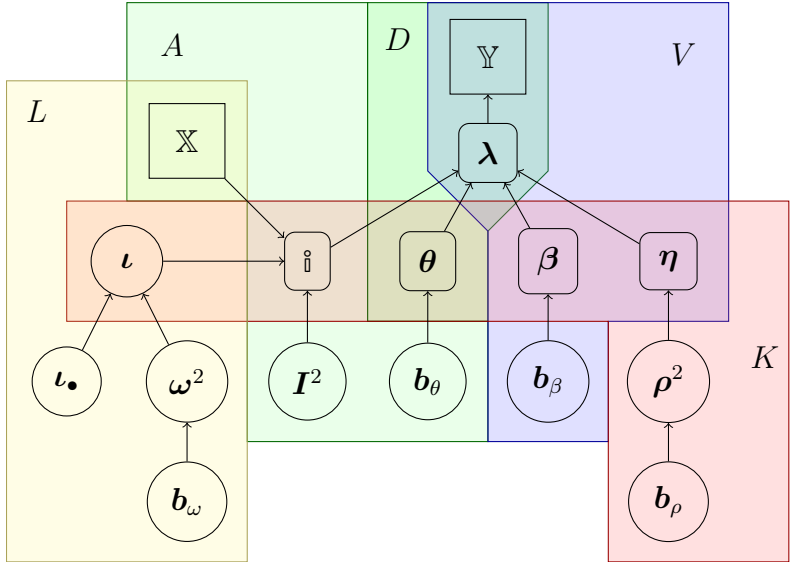
\begin{figure}[!t]
	\centering
    \begin{tikzpicture}[x=0.8cm,y=0.8cm](13.2,9.3)
        \linethickness{0.075mm}
        % borders of the picture
        %\draw (7,-5) -- (7,-5) -- (7,5) -- (-7,5) -- cycle;
        \clip (-8.1,-5) rectangle (5.1,4.3);
        %% Plates - D, A, V, K, L
        \filldraw[fill=green!80!white, draw=green!40!black, fill opacity=0.2] (-2,4.3) -- (1,4.3) -- (1,1.5) -- (0,0.5) -- (0,-1) -- (-2,-1) -- cycle;
        \filldraw[fill=green!40!white, draw=green!40!black, fill opacity=0.2] (-2,4.3) -- (-2,-1) -- (0,-1) -- (0,-3) -- (-4,-3) -- (-4,-1) -- (-4,01) -- (-6,1) -- (-6,4.3) -- cycle;
        \filldraw[fill=blue!60!white, draw=blue!60!black, fill opacity=0.2] (4,4.3) -- (-1,4.3) -- (-1,1.5) -- (0,0.5) -- (0,-1) -- (0,-3) -- (2,-3) -- (2,-1) -- (4,-1) -- cycle;
        \filldraw[fill=red!60!white, draw=red!60!black, fill opacity=0.2] (-7,1) -- (5, 1) -- (5,-5) -- (2,-5) -- (2,-1) -- (-7,-1) -- cycle;
        \filldraw[fill=yellow!60!white, draw=yellow!60!black, fill opacity=0.2] (-8,3) -- (-4,3) -- (-4,-5) -- (-8,-5) -- cycle;
        %% Labels of the plates
        \draw (-1.5, 3.7) node {$D$};
        \draw (-5.25, 3.6) node {$A$};
        \draw (3.25, 3.45) node {$V$};
        \draw (4.6, -1.6) node {$K$};
        \draw (-7.5, 2.5) node {$L$};
        %% Nodes
        % outcomes Y + regressors X
        \draw (0,3.4) node(Y)[draw = black, rectangle, inner sep = 10pt] {$\mathbb{Y}$};
        \draw (-5,2) node(X)[draw = black, rectangle, inner sep = 10pt] {$\mathbb{X}$};
        % Poisson intensities lambda
        \draw (0,1.8) node(lambda)[draw = black, rectangle, rounded corners, inner sep = 7 pt] {$\bm \lambda$};
        % Primary parameters
        \draw (-3,0) node(ideal)[draw = black, rectangle, rounded corners, inner sep = 7 pt] {$\bm \ideal$};
        \draw (-1,0) node(theta)[draw = black, rectangle, rounded corners, inner sep = 7 pt] {$\bm \theta$};
        \draw (1,0) node(beta)[draw = black, rectangle, rounded corners, inner sep = 7 pt] {$\bm \beta$};
        \draw (3,0) node(eta)[draw = black, rectangle, rounded corners, inner sep = 7 pt] {$\bm \eta$};
        % Other auxiliary parameters
        \draw (-6,0) node(iota)[draw = black, circle, inner sep = 7pt] {$\bm \iota$};
        \draw (-3,-2) node(I)[draw = black, circle, inner sep = 5pt] {$\bm I^2$};
        \draw (-1,-2) node(btheta)[draw = black, circle, inner sep = 5pt] {$\bm b_\theta$};
        \draw (1,-2) node(bbeta)[draw = black, circle, inner sep = 5pt] {$\bm b_\beta$};
        \draw (3,-2) node(rho)[draw = black, circle, inner sep = 5pt] {$\bm \rho^2$};
        \draw (3,-4) node(brho)[draw = black, circle, inner sep = 5pt] {$\bm b_\rho$};
        \draw (-7,-2) node(iotamean)[draw = black, circle, inner sep = 5pt] {$\bm \iota_\bullet$};
        \draw (-5,-2) node(omega)[draw = black, circle, inner sep = 5pt] {$\bm \omega^2$};
        \draw (-5,-4) node(bomega)[draw = black, circle, inner sep = 5pt] {$\bm b_\omega$};
        %% Directions
        \draw [->] (lambda) -- (Y);
        \draw [->] (ideal) -- (lambda);
        \draw [->] (theta) -- (lambda);
        \draw [->] (beta) -- (lambda);
        \draw [->] (eta) -- (lambda);
        \draw [->] (iota) -- (ideal);
        \draw [->] (I) -- (ideal);
        \draw [->] (X) -- (ideal);
        \draw [->] (btheta) -- (theta);
        \draw [->] (bbeta) -- (beta);
        \draw [->] (rho) -- (eta);
        \draw [->] (iotamean) -- (iota);
        \draw [->] (omega) -- (iota);
        \draw [->] (bomega) -- (omega);
        \draw [->] (brho) -- (rho);
    \end{tikzpicture}
    \caption{\label{fig:hierarchy_STBSM} 
		The STBS model. Observed data (rectangles), latent variables of interest (rounded corners), additional parameters (circles). Colored planes indicate the relationship to documents ($D$), authors ($A$), the vocabulary ($V$), topics ($K$) and covariates ($L$).
	}
\end{figure} 

%% file: sec03_inference.tex
\section{Inference} \label{sec:inference}

The STBS model contains the following latent variables of interest: topic intensities $\bm \theta$, neutral word intensities $\bm \beta$, polarity values $\bm \eta$ and ideological positions $\bm \ideal$. In addition, we have the following parameters: document(author)-specific prior rates $\bm b_\theta$ for $\bm \theta$, term-specific prior rates $\bm b_\beta$ for $\bm \beta$, topic-specific precisions $\bm \rho^2$ for polarity values $\bm \eta$ and their prior rates $\bm b_\rho$, topic-specific regression coefficients $\bm \iota$ for $\bm \ideal$, author-specific precisions $\bm I^2$ of the regression model, prior means for each regression coefficient $\iota_{\bullet}$ and their precisions $\bm \omega^2$ with prior rates $\bm b_{\omega }$. Altogether, the posterior distribution of interest is given by $p( \bm \zeta | \mathbb{Y}, \mathbb{X}) = p(\bm \theta, \bm \beta, \bm \eta, \bm \ideal, \bm b_\theta, \bm b_\beta, \bm \rho^2, \bm b_\rho, \bm \iota, \bm I^2, \bm \iota_{\bullet}, \bm \omega^2, \bm b_{\omega} | \mathbb{Y}, \mathbb{X})$. 

Due to a mostly conditionally conjugate prior specification, Markov chain Monte Carlo sampling based on Metropolis within Gibbs sampling could be utilized to approximate the posterior distribution if it were not for the high dimensionality of the latent variables and the model parameters. The high dimensionality implies that a sampling-based approach is prohibitively expensive for a large corpus, both time and memory-wise. Hence, we use variational inference to approximate the posterior distribution, similar to previous work on this type of models. In the following, we present the variational inference scheme used including the variational families, discuss the combination of stochastic gradient ascent with coordinate ascent used in the optimization algorithm, provide the post-processing and final inference step and outline implementation details.

%%%!!!!!!!!!!!!!!!!!!!!!!!!!!!!!!!!!!!!!!!!!!!!!!!!!!!!!!!!!!!!!!!!!!!!!!!%%%
%%%%%%%%%%%%%%%%%%%%%%%%%%%%%%%%%%%%%%%%%%%%%%%%%%%%%%%%%%%%%%%%%%%%%%%%%%%%%
%%%%%%%%%%%%               NEW SUBSECTION                   %%%%%%%%%%%%%%%%%
%%%%%%%%%%%%%%%%%%%%%%%%%%%%%%%%%%%%%%%%%%%%%%%%%%%%%%%%%%%%%%%%%%%%%%%%%%%%%
%%%!!!!!!!!!!!!!!!!!!!!!!!!!!!!!!!!!!!!!!!!!!!!!!!!!!!!!!!!!!!!!!!!!!!!!!!%%%
\subsection{Variational inference} \label{subsec:variational_inference}

The variational approach transfers the inference problem into an optimization problem. One approximates the posterior distribution by a flexible variational family of distributions, i.e., one aims at finding the set of variational parameters $\bm \phi \in \bm \Phi$ such that the corresponding variational distribution $q_{\bm \phi}$ resembles the posterior distribution the most. In particular, one aims to determine $\bm \phi$ minimizing the Kullback-Leibler divergence of $q_{\bm \phi}$ from the posterior. 

Equivalently, this problem is solved by maximizing the evidence lower bound ($\mathsf{ELBO}$):
\begin{equation} \label{eq:ELBO}
    \ELBO{\bm \phi} = \Eq \left[ \log p(\mathbb{Y} | \bm \zeta) + \log p (\bm \zeta | \mathbb{X}) + (-\log q_{\bm \phi} (\bm \zeta)) \right],
\end{equation}
where $\Eq$ is the expectation with respect to the variational distribution with parameter $\bm \phi$. The three summands are referred to as \emph{reconstruction}, \emph{log-prior} and  \emph{entropy} of the variational distribution $q_{\bm \phi}$.   

As variational family, we consider the mean-field family which assumes independence between all parameters and latent variables. For each parameter or latent variable, a~distributional family described by two parameters is selected that aligns with the prior distribution for this variable. We choose either normal or gamma families since their conjugacy yields direct updates for the coordinate ascent algorithm. 
% This is different from \citet{Vafa_etal_2020} where log-normal densities were used for positive latent variables and parameters. However, \citet{Vafa_etal_2020} did also not use direct updates based on the coordinate ascent algorithm.\footnote{Their updated implementation of the TBIP model available on Github~\citep{Github_TBIP} also combines SVI with CAVI updates facilitated by the use of gamma priors.
We impose a multivariate normal distribution as variational family with a general covariance matrix on the regression coefficients to relax the independence assumption.
In particular, we posit the following independent variational families:
\begin{equation} \label{eq:meanfield_var_family}
	\begin{split}
		\begin{aligned}
			q(\theta_{dk}) &= \gammadist{\phi_{\theta dk}^\shp}{\phi_{\theta dk}^\rte}, 
			\quad&
			q(\beta_{kv}) &= \gammadist{\phi_{\beta kv}^\shp}{\phi_{\beta kv}^\rte},
			\quad&
			q(\eta_{kv}) &= \norm{\phi_{\eta kv}^\loc}{\phi_{\eta kv}^\sclsq},
			\\
			q(\ideal_{ak}) &= \norm{\phi_{\ideal ak}^\loc}{\phi_{\ideal ak}^\sclsq},
            \quad&
            q(b_{\theta a}) &= \gammadist{\phi_{b \theta a}^\shp}{\phi_{b \theta a}^\rte},
            \quad&
            q(b_{\beta v}) &= \gammadist{\phi_{b \beta v}^\shp}{\phi_{b \beta v}^\rte},
            \\
            q(\rho_k^2) &= \gammadist{\phi_{\rho k}^\shp}{\phi_{\rho k}^\rte}, 
            \quad&
            q(\omega_l^2) &= \gammadist{\phi_{\omega l}^\shp}{\phi_{\omega l}^\rte},
            \quad&
            q(b_{\rho k}) &= \gammadist{\phi_{b\rho k}^\shp}{\phi_{b\rho k}^\rte}, 
            \\
            q(b_{\omega l}) &= \gammadist{\phi_{b\omega l}^\shp}{\phi_{b\omega l}^\rte},
            \quad&
            q(\bm \iota_{k}) &= \knorm{L}{\bm \phi_{\iota k}^\loc}{\bm \phi_{\iota k}^\covm}, 
            \quad&
            q(\bm \iota_{\bullet}) &= \knorm{L}{\bm \phi_{\iota\bullet}^\loc}{\bm \phi_{\iota\bullet}^\covm},
            \\
            &\quad&
            q(I_a^2) &= \gammadist{\phi_{Ia}^\shp}{\phi_{Ia}^\rte},
		\end{aligned}
	\end{split}
\end{equation}
where $\phi_\bullet^\shp$ and $\phi_\bullet^\rte$ are the shape and rate parameters of the gamma distribution, $\phi_\bullet^\loc$, $\phi_\bullet^\sclsq$ and $\bm \phi_\bullet^\covm$ are the location, variance and covariance matrix parameters of the (multivariate) normal distribution. 

%%%!!!!!!!!!!!!!!!!!!!!!!!!!!!!!!!!!!!!!!!!!!!!!!!!!!!!!!!!!!!!!!!!!!!!!!!%%%
%%%%%%%%%%%%%%%%%%%%%%%%%%%%%%%%%%%%%%%%%%%%%%%%%%%%%%%%%%%%%%%%%%%%%%%%%%%%%
%%%%%%%%%%%%               NEW SUBSECTION                   %%%%%%%%%%%%%%%%%
%%%%%%%%%%%%%%%%%%%%%%%%%%%%%%%%%%%%%%%%%%%%%%%%%%%%%%%%%%%%%%%%%%%%%%%%%%%%%
%%%!!!!!!!!!!!!!!!!!!!!!!!!!!!!!!!!!!!!!!!!!!!!!!!!!!!!!!!!!!!!!!!!!!!!!!!%%%
\subsection{Stochastic gradient ascent within coordinate ascent} \label{subsec:SVI_within_CAVI}

Our inference problem has been cast as an optimization problem for $\ELBO{\bm \phi}$ w.r.t.\ $\bm\phi$. We solve this optimization problem using a stochastic automatic differentiation gradient ascent within coordinate ascent procedure, i.e., we combine ideas from stochastic variational inference \citep[SVI;][]{Hoffman+Blei+Wang:2013},
automatic differentiation variational inference \citep[ADVI;][]{Kucukelbir+Tran+Ranganath:2017} and the optimization scheme employed in \citet{Vafa_etal_2020}. We pursue a stochastic gradient ascent approach where we split the complete parameter vector into blocks and where one block is updated by moving it in the direction of a noisy gradient while keeping the remaining parameters fixed making use of data subsampling. 

The update step varies depending on if the step is analytically available in closed form or needs to be obtained using automatic differentiation for the gradient. In the case, where the step is analytically available in closed form, we make use of SVI which is based on coordinate ascent variational inference (CAVI).
Due to well-formulated conjugacy, the CAVI approach is feasible for the parameters in the STBS model which correspond to parameters in the hierarchical Poisson factorization model \citep{gopalanhofmanblei2015scalableHPoisF} and for the regression coefficients. CAVI update steps are not available in closed form under non-conjugacy. This is the case in our model for the ideological factor $\exp\{\eta_{kv}\ideal_{ak}\}$ in the Poisson rates in~\eqref{eq:Pois_counts}. Hence, we employ a variant of ADVI for $\ideal$ and $\bm{\eta}$, i.e., gradients are determined based on automatic differentiation in combination with the Adam optimizer \citep{kingma_ba2015Adam}. Also the expectation of the ideological factor -- which is required for the CAVI updates of other parameters -- deserves a~special treatment, with more details given in Appendix~\ref{app:expectedterm} and \ref{app:stochastic-gradients}.

Due to the high dimensionality of the full parameter vector, we prefer such a differentiated approach to using a fully automatic procedure which would ignore the availability of CAVI updates.
% Most of the parameters are updated directly based on CAVI. 
Some of these updates are direct generalizations of updates in \citet{gopalanhofmanblei2015scalableHPoisF}, the updates for the shrinkage triple gamma prior follow naturally from our conjugate formulation~\eqref{eq:prior_eta}, see Appendix~\ref{app:cavi} for details. 
%The remaining parameters are updated using ADVI with the Adam optimizer \citep{kingma_ba2015Adam}. %as for the original TBIP model \citep{Vafa_etal_2020}.
%Similar compilation has already appeared in the literature for different models (references).

We estimate the hierarchical Poisson factorization model with the same number of topics $K$ to obtain initial estimates $\bm \phi_\theta^0$ and $\bm \phi_\beta^0$. Ideological positions are initialized using reasonable values to fix the meaning of the polarity, e.g., for the U.S.\ Senate speeches $\ideal_{ak} \in \{-1,0,1\}$ depending on the political party of the author~$a$. Other parameters are initialized randomly, see \citet{Vafa_etal_2020}. Then, we perform parameter updates for $E$ epochs. In each epoch, the documents are randomly split into batches of pre-specified size $|\mathcal{B}|$ ensuring that all documents contribute to the construction of stochastic gradients in each epoch. Given a~batch $\mathcal{B}$, only documents from this batch are used for updates of local and global variables and $\ELBO{\bm\phi}$ evaluation. This implies that the sums $\sum\limits_{d=1}^D$ are replaced by $\sum\limits_{d \in \mathcal{B}_b}$ and scaled by $\frac{D}{|\mathcal{B}_b|}$. As step size for the global CAVI updates, we choose $\rho_t = (t+\tau)^{-\kappa}$. This satisfies the Robbins-Monro condition if delay $\tau\ge 0$ and exponent $\kappa \in (0.5, 1]$. 

\begin{algorithm}[!t]
	\caption{Estimation combining CAVI and SVI.}
	\label{alg:TBIP_CAVI_and_SVI}
	\begin{algorithmic} [1]
		\State \textbf{Input:} $\mathbb{Y}$, $\mathbb{X}$, initial value $\bm \phi^0$, $E$, $|\mathcal{B}|$, $I$, $\kappa \in (0.5, 1]$, $\tau \geq 0$, $\alpha$, hyperparameter values.
        \State Initialize the step counter $t:= 0$.
		\For{$e$ in $1:E$}
		\State \textbullet~Divide $D$ documents randomly into $B$ batches $\mathcal{B}_b, b = 1, \ldots, B$ of size $|\mathcal{B}_b| \approx |\mathcal{B}|$.
            \For{$b$ in $1:B$}
            %\State \textbullet~Only the documents $d\in\mathcal{B}_b$ are available.
            %\State \textbullet~Use sums $\sum\limits_{d \in \mathcal{B}_b} \cdots$ scaled by $\frac{D}{|\mathcal{B}_b|}$ instead of $\sum\limits_{d=1}^D \cdots$.
            \State \textbullet~Set $t:= t+1$ and the step size $\rho_t = (t+\tau)^{-\kappa}$. 
            \State \textbullet~Update $\bm \phi_{\theta dk}^t := \widehat{\bm \phi}_{\theta dk}$ for $d \in \mathcal{B}_b$. 
            \Comment{CAVI \emph{local} updates}
                \For{$\phi$ in $\{\bm \phi_{b\theta}, \bm \phi_\beta, \bm \phi_{b\beta}, \bm \phi_\rho, \bm \phi_{b\rho}, \bm \phi_{\iota}, \bm \phi_{\iota\bullet}, \bm \phi_I, \bm \phi_\omega, \bm \phi_{b\omega} \}$}
                \State \textbullet~Update $\phi^t := \rho_t \widehat{\phi} + (1-\rho_t) \phi^{t-1} $. 
                \Comment{CAVI \emph{global} updates}
                \EndFor
            \State \textbullet~Approximate $\ELBO{\bm \phi}$ by replacing $\Eq f(\bm \zeta)$ with $\frac{1}{I} \sum\limits_{i=1}^I f(\bm \zeta^i)$ where $\bm\zeta^i \overset{\text{iid}}{\sim} q_{\bm \phi}$. 
            \State \textbullet~Track noisy gradients using reparametrization trick for $\bm \phi_\eta$ and $\bm \phi_\ideal$. 
            \State \textbullet~Update $\bm \phi_\eta$ and $\bm \phi_\ideal$ with Adam and learning rate $\alpha$. 
            \Comment{ADVI \emph{global} updates}
            \EndFor
		\EndFor
	\end{algorithmic}
\end{algorithm}

Algorithm~\ref{alg:TBIP_CAVI_and_SVI} summarizes the steps of our stochastic gradient ascent within coordinate ascent procedure. To simplify notation, $\widehat{\bm \phi}_\bullet$ denotes the CAVI update for parameter $\bm \phi_\bullet$ when all other parameters are fixed, for details see Appendix~\ref{app:cavi}. Note that variational parameters for $\bm \theta$ are updated directly without any convex combination with the previous value as they are local parameters. Other parameters are global since they are not document-specific. Updates of $\widehat{\bm \phi}_\beta$ and $\widehat{\bm \phi}_{b\theta}$ depend on $\bm \phi_\theta$ but only the subset from the current batch is used after appropriate scaling by $\frac{D}{|\mathcal{B}_b|}$ to save computational time. 
% Additional speed-ups can be achieved by using a suitable sequence of CAVI updates. 

\begin{comment}
\begin{enumerate}
    \item First, suggest full-stochastic gradient approach (noisy gradients, MC reparametrization trick, Adam \ldots all as in \citep{Vafa_etal_2020}). 
    \item Objection about too complicated to track gradients for all the parameters.
    \item Hierarchical priors would allow for full-condition distributions to create Gibbs sampling. 
    \item Similarly, we can help ourselves with coordinate ascent.
    \item Hence, CAVI updates where it can be performed, see Appendix~\ref{sec:implementation-details} for details.
    \item \textcolor{red}{Do a research on combining stochastic gradient approach with CAVI updates.}
    \item And stochastic gradients for parameters that cannot be directly CAVI updated, see Appendix~\ref{sec:implementation-details} for details.
    \item Correcting \citep{Vafa_etal_2020} - expectation of the ideological term here or just in the Appendix~\ref{sec:implementation-details}?
\end{enumerate}
\end{comment}

%%%!!!!!!!!!!!!!!!!!!!!!!!!!!!!!!!!!!!!!!!!!!!!!!!!!!!!!!!!!!!!!!!!!!!!!!!%%%
%%%%%%%%%%%%%%%%%%%%%%%%%%%%%%%%%%%%%%%%%%%%%%%%%%%%%%%%%%%%%%%%%%%%%%%%%%%%%
%%%%%%%%%%%%               NEW SUBSECTION                   %%%%%%%%%%%%%%%%%
%%%%%%%%%%%%%%%%%%%%%%%%%%%%%%%%%%%%%%%%%%%%%%%%%%%%%%%%%%%%%%%%%%%%%%%%%%%%%
%%%!!!!!!!!!!!!!!!!!!!!!!!!!!!!!!!!!!!!!!!!!!!!!!!!!!!!!!!!!!!!!!!!!!!!!!!%%%
\subsection{Post-processing and final inference} \label{subsec:inference}
After model fitting we summarize the results by determining point estimates for the parameters based on posterior means. These point estimates are used to inspect the topic-specific polarity, assess the regression results to characterize the ideological positions, and obtain aggregate author-specific ideological positions.  Prevalent terms for each topic are identified using a plug-in estimator based on these point estimates to obtain corrected intensities on a log scale. The inspection of regression results makes use of highest posterior density (HPD) intervals to determine closeness to a null model and interpretation is eased using a convenient visualization of the results.

\subsubsection*{Posterior estimates}
Having confirmed convergence of the ELBO values after a sufficient number of epochs $E$, the last values obtained for the variational parameters represent the final estimates $\widetilde{\bm \phi} = \bm \phi^{\mathit{BE}}$. These variational parameters describe the estimated variational families that are closest to the posterior distribution in terms of Kullback-Leibler divergence. We use these estimates to determine point estimates for the parameters of the posterior distribution of interest. In particular we determine the posterior mean estimates of $\bm{\zeta}$ based on the means induced by the variational distributions: $\widetilde{\phi}_{\bullet}^\loc$ for the normal variational family, $\widetilde{\phi}_{\bullet}^\shp / \widetilde{\phi}_{\bullet}^\rte$ for the gamma variational family. 

HPD intervals for the parameters are derived based on the quantiles of the normal and gamma variational densities. For the regression coefficients $\bm \iota_k$, the posterior is approximated by a multivariate normal distribution given by
$
\bm \iota_k | \mathbb{Y}, \mathbb{X} 
% https://stats.stackexchange.com/questions/183046/what-exactly-does-dot-sim-notation-mean
\mathrel{\dot\sim}
\knorm{L}{\widetilde{\bm \phi}_{\iota k}^\loc}{\widetilde{\bm \phi}_{\iota k}^\covm}
$ for each topic~$k$. Moreover, the posterior distribution of any linear combination of regression coefficients induced by a matrix~$\mathbb{C} \in \R^{C\times L}$ is approximated by 
$
\left.
    \mathbb{C} \left( \bm \iota_k - \widetilde{\bm \phi}_{\iota k}^\loc\right) 
\right| \mathbb{Y}, \mathbb{X} 
\mathrel{\dot\sim}
\knorm{C}{\bm 0}{\mathbb{C} \widetilde{\bm \phi}_{\iota k}^\covm \; \T{\mathbb{C}}}
$. Hence, HPD regions are obtained via 
$$
\left.
    \T{\left(\bm \iota_k - \widetilde{\bm \phi}_{\iota k}^\loc\right)} \T{\mathbb{C}}
    \left[\mathbb{C}\; \widetilde{\bm \phi}_{\iota k}^\covm \; \T{\mathbb{C}}\right]^{-1} 
    \mathbb{C} \left(\bm \iota_k - \widetilde{\bm \phi}_{\iota k}^\loc\right) 
\right| \mathbb{Y}, \mathbb{X} 
\mathrel{\dot\sim}
\chi_{C}^2
$$
provided the inverse matrix exists. 

To evaluate the closeness of the approximated posterior mean to $\bm 0$, we construct a HPD region that touches $\bm 0$ (i.e., $\bm 0$ lies on the edge of the region) and compute the complementary coverage probability (CCP), i.e., one minus the probability covered by this HPD region. %Large values of CCP correspond to $\bm 0$ being close to the center of the posterior probability mass. While low values close to 0 suggest effects significantly distant from zero (no effect in general). 
When analyzing the model results later, we use the following labels to indicate small values of the CCP: $***$, $**$, $*$, $\cdot$ for values below 0.001, 0.01, 0.05, 0.1, respectively. In particular, we determine CCP values for single regression coefficients as well as for the set of regression coefficients corresponding to the same categorical variable and use them to infer the relevance of these coefficients and variable to characterize the ideological positions.

Author-specific ideological positions $\ideal_{ak}$ are estimated by $\phi_{ak}^\loc$ for each topic~$k$ separately. In order to compare aggregate author-specific ideological positions to the ideal points obtained for the TBIP model, we take a weighted average across the ideological positions over all topics. To respect the prevalence of each topic, we assign to the ideological positions of highly discussed topics a higher weight. In particular, the weights are constructed from the variational means of the topic intensities $\theta_{dk}$ averaged for documents written by the same author in the following way:
$
w_{ak} = \frac{1}{|\mathcal{D}_a|} \sum \limits_{d \in \mathcal{D}_a} \widetilde{\phi}_{\theta dk}^\shp \,/\, \widetilde{\phi}_{\theta dk}^\rte.
$

We assess which terms are highly prevalent when discussing a certain topic~$k$ with different ideological position $\mathsf{i} \in \{-1, 0, 1\}$ by determining a plug-in estimate based on the posterior mean estimates for $\bm \beta$ and $\bm \eta$ and obtaining ideology-corrected intensities on a log scale using:
$$
\Eq \log \left( \beta_{kv} \exp \left\{\eta_{kv} \mathsf{i} \right\} \right)
=
\Eq \log \left( \beta_{kv} \right) + \mathsf{i} \Eq \eta_{kv}
=
\psi\left(\widetilde{\phi}_{\beta kv}^\shp\right) - \log \left(\widetilde{\phi}_{\beta kv}^\rte\right) + \mathsf{i} \cdot \widetilde{\phi}_{\eta kv}^\loc,
$$
where $\psi(\cdot)$ is the digamma function. Using the log scale implies that the result could be negative. To avoid negative values, we shift the result by subtracting the minimum and adding 5\% of the range. We order the terms according to these ideology-corrected intensities and construct word clouds of the 10 most frequently used terms for each topic~$k$.

\subsubsection*{Visualizing regression results}

To summarize the regression results, we create a \emph{regression summary plot}. This plot provides an overview on the estimated ideological positions, potentially grouped by a covariate, using histograms as well as the estimated regression coefficients to assess the impact of the covariates on ideological positions. The plot differs depending on if the model includes only main effects or also interaction effects between one main variable and the remaining variables. This plot assumes that all regressors are categorical.

If the regression model only includes main effects, the plot shows a histogram of the estimated ideological positions on the left, potentially grouped and colored indicating the group-specific average (see Figure~\ref{fig:TBIPhier_ideal_a_Nreg_all_no_int_covariate_effects}). The estimated regression coefficients are then column-wise shown for each of the covariates included in the regression. For each category of the covariate, a rectangle is included where the height is proportional to the number of observations. Next to each rectangle, a text label is inserted indicating the name of the category and the number of observations. The area of the rectangles is colored by the value of the regression coefficient (with a color legend on the right) and stars are added corresponding to the CCP values. The CCP value obtained when the regression coefficients of all categories of the covariate are set to zero is given at the bottom for each covariate. 

If the regression model also includes an interaction effect between one main covariate and all other covariates, the regression summary plot is changed by including row-wise the separate results for the categories of the main covariate (see Figure~\ref{fig:TBIPhier_all_party_effects_interactions_transposed_k_9}). The effects of the main covariate (when the other covariates are at their baseline levels) are visualized in the first column next to the histogram. The rectangles are included for each of the categories of the main covariate with height corresponding to the number of observations. The area of the rectangles is colored according to the values of the regression coefficients and stars corresponding to the CCP values are added inside the rectangles. At the bottom the CCP value is included which assesses if these effects differ from the baseline category. 

Next, the interaction effects of the main covariate with the other covariates are column-wise included. The color of the rectangle area indicates the effect the category of this covariate has in comparison to the baseline level conditional on the main covariate having the row-specific value. The stars inside the rectangles are based on the corresponding CCP values. At the bottom, the CCP value is shown which is obtained when all interaction effects are dropped, i.e., when the coefficients are restricted to correspond only to a main effect model for this covariate. 

%\begin{itemize}
 %   \item[$\pm$] Add the description of finding the most influential speeches.
%\end{itemize}

%%%!!!!!!!!!!!!!!!!!!!!!!!!!!!!!!!!!!!!!!!!!!!!!!!!!!!!!!!!!!!!!!!!!!!!!!!%%%
%%%%%%%%%%%%%%%%%%%%%%%%%%%%%%%%%%%%%%%%%%%%%%%%%%%%%%%%%%%%%%%%%%%%%%%%%%%%%
%%%%%%%%%%%%               NEW SUBSECTION                   %%%%%%%%%%%%%%%%%
%%%%%%%%%%%%%%%%%%%%%%%%%%%%%%%%%%%%%%%%%%%%%%%%%%%%%%%%%%%%%%%%%%%%%%%%%%%%%
%%%!!!!!!!!!!!!!!!!!!!!!!!!!!!!!!!!!!!!!!!!!!!!!!!!!!!!!!!!!!!!!!!!!!!!!!!%%%
\subsection{Implementation} \label{subsec:implementation}

Model fitting is implemented in TensorFlow 2.0 %\citep{tensorflow2015} 
which is designed for gradient based optimization of unrestricted variables as well as TensorFlow's add-on library for probabilistic reasoning, TensorFlow Probability. Conditions such as positivity (or even the structure of a~covariance matrix) are incorporated by suitable (chains of) \texttt{tfp.bijectors}, e.g., the \emph{softplus} function $\log(1+\exp\{x\})$. Automatic splitting of the workload into batches is straightforward within this environment, but also required because objects of size $D \times K \times V$ easily exceed the memory limits unlike objects of size $|\mathcal{B}| \times K \times V$. Function \texttt{GradientTape} is able to track the required gradients during $\ELBO{\bm\phi}$ approximation including the implicit reparametrization trick. The gradients are combined with one of the implemented \texttt{tf.optimizers}, i.e., Adam, to perform the gradient ascent update. The CAVI updates are manually added and inserted before the $\ELBO{\bm\phi}$ approximation.
The model manager allows to quickly save an intermediate model configuration including current parameter values. This intermediate model can then be promptly loaded and estimation continued. 

For the empirical evaluation of the STBS model, we use in particular the following choices for the hyperparameters of the hierarchical prior structure:
\begin{align*}
        a_\theta &= 0.3,&a'_\theta&= 0.3,&b'_\theta&=0.3,&
        a_\beta &= 0.3,&a'_\beta&= 0.3,&b'_\beta&=0.3,\\
        &&a_{\rho}&=0.3,&a'_{\rho}&=0.3,&\kappa^2_{\rho}&=10,\\
        &&a_{\omega}&=0.3,&a'_{\omega}&=0.3,&\kappa^2_{\omega}&=10,&
        a_I &=0.3,&b_I&=0.3.
\end{align*}
Additional values specified when running the algorithm are:
$E = 1,000$, $|\mathcal{B}|=512$, $I = 1$, $\kappa =0.51$, $\tau = 0$, $\alpha = 0.01$. Convergence of the algorithm is visually assessed by inspecting the trace plot of the ELBO values versus epochs. $1,000$ epochs are clearly sufficient for convergence.
Our flexible implementation that allows for several different prior structure choices is available at Github \citep{Github_STBS}.

%% file: sec04_senate_Paul.tex
\section{The unfolding of ideology in political discourse} \label{sec:senate}
 
In the following, we use our model to gain insights into the influence of ideology on the political discourse through speeches delivered in the 114th U.S.\ Senate (2015--2017). We use the data set provided by \citet{Gentzkow+Shapiro+Taddy:2018} and previously analyzed in, e.g., \cite{Taddy_etal_2019, RSS_2023, Hofmarcher+Adhikari+Gruen:2022, Vafa_etal_2020}. We fit two different variants of the STBS model: (1) with ideological positions fixed over topics and (2) with topic-specific ideological positions. We determine and compare the induced polarity for each topic for both models. For the model with fixed ideological positions, we highlight the additional insights provided by the inclusion of speaker-specific covariates such as gender, religion, or experience to characterize the ideological positions.  Inspecting the model with topic-specific ideological positions allows us to indicate how this additional flexibility changes results and enables a more fine-grained analysis.

In the first STBS model variant, we assume that the ideological positions
$\ideal_{a}$ are identical across topics for each speaker and characterize them using author-specific covariates. We use an additive regression setup for the ideological positions, incorporating the following categorical covariates: party (baseline category: Democrat), gender (Male), region (Northeast), generation (Silent), experience (10+ years in Congress), and religion (Other), i.e., the regression formula for the ideological positions is given by:
\begin{equation}
\label{eq:example1}
\texttt{$\sim$ party + gender + region + generation + experience + religion}
\end{equation}
Subsequently, in the second variant, we allow for topic-specific ideological positions $\ideal_{ak}$. Additionally, we introduce interaction effects between party and the remaining covariates, enhancing our model's granularity in understanding both ideological positions and the impact of covariates at the party level. This extended regression model is represented by:
\begin{equation}
\label{eq:example2}
\texttt{$\sim$ party * (gender + region + generation + experience + religion)}
\end{equation}
In the following, we describe the data set and the data pre-processing in Section~\ref{sec:senate:data}. We then compare the polarity in language on topic level for both model variants in Section~\ref{subsec:res_pol}. Sections~\ref{subsec:reg_ideal_a} and \ref{subsec:detailed_regression} present the detailed regression results to characterize the influence of the covariates on ideology for both models.

\subsection{Speech data, covariates and pre-processing} \label{sec:senate:data}
The Congress Record provided by \citet{Gentzkow+Shapiro+Taddy:2018} includes speech data for the 114th Congress session in the Senate. In addition, this data source also contains information on gender (``Gender'') and for which state (``Region'') Congress members were elected. We add information on the age (``Generation'') and  years of service for each Congress member \citep[``Experience''; obtained via][]{Skelley2023, Github_538}, but also religious affiliation (``Religion''; \citealt{PewResearchCenter}). The categories defined for these variables can be found in Appendix~\ref{app:overview}. Summary statistics of the demographics of the Senators used in the analysis are provided by Figure~\ref{fig:TBIPhier_ideal_a_Nreg_all_no_int_covariate_effects} in Section~\ref{subsec:reg_ideal_a}, which visualizes regression results. 

Pre-processing the speech data to obtain a document-term matrix follows essentially \cite{Gentzkow_etal_2019} and \citet{Vafa_etal_2020}. In particular, we remove Senators with less than 24 speeches resulting in 99 Senators. 
Following \cite{Gentzkow_etal_2019}, we use bigrams for our analysis. We use the same list of stop words as \citet{Vafa_etal_2020} and  only keep the bigrams that appear in at least 0.1\% and at most 30\% of the documents.\footnote{In particular, the pre-processing procedure utilizes the \texttt{CountVectorizer} function from the \textbf{scikit-learn} package.}% \citep{scikit-learn}.}
Finally, we restrict the analysis to bigrams used by 10 or more Senators to eliminate terms used only by a~small number of Senators. The pre-processing results in $D=14\,672$ non-empty speeches containing $V=5\,031$ terms. The choice of the number of topics follows \citet{Hofmarcher+Adhikari+Gruen:2022} who add 3 topics to the 22 topics manually predefined in 
\citet{Gentzkow_etal_2019} to capture the content in Congress speeches in a suitable way. We initialize the variational parameters $\phi_\theta$ and $\phi_\beta$ with the results obtained for the hierarchical Poisson factorization model with $K=25$. This essentially preserves the meaning of each of the topics $k \in \{1, \ldots, K\}$ across the two fitted model variants. 

\subsection{Topic polarity}
\label{subsec:res_pol}
The STBS model enables the detection of polarity in speech at topic level. Polarity is quantified as the variability of the product of ideological positions and polarity values, denoted by $\eta_{kv}\ideal_{ak}$ (see Equation~\eqref{eq:Pois_counts}). Figure~\ref{fig:eta_ideal_variability_a_vs_ak} shows the polarity estimates with the green bars representing the results where ideological positions are fixed across topics and the orange bars depicting the results where ideological positions are allowed to vary across topics. 
To ease interpretation of results, Figure~\ref{fig:eta_ideal_variability_a_vs_ak} also includes labels for all topics which we obtained by inspecting the most frequent terms associated with each topic and reviewing the speeches assigned to the topic.
Figure~\ref{fig:eta_ideal_variability_a_vs_ak} shows the following: Firstly, the more flexible model, employing topic-specific ideological positions, yields higher polarity estimates across all topics. 
Secondly, polarity varies among topics regardless of if the ideological positions are fixed or allowed to vary across topics. Some congruence between the extent of polarity is discernible for the two models. For example, topic~24 (Export, Import and Business) clearly has the highest polarity score for the model with topic-specific ideological positions but is also among the topics with highest polarity scores for the fixed positions model. Topic~4 (Commemoration and Anniversaries) has the lowest polarity for the model with topic-specific ideological positions and also has a rather low polarity score for the fixed positions model. By contrast, there are also topics where the extent of polarity between the two models varies considerable. The highest difference in the extent of polarity between the two model specifications is observed for topic~9 (Veterans and Health Care), where for the model with fixed ideological position hardly any polarity in speech is inferred while the polarity score is among the higher ones for the more flexible model. 
These results imply that overall a similar picture emerges for both model specifications. However, the more flexible specification captures more variability due to polarity.

%However, we aim not only to detect polarity but also to understand its underlying reasons and the extent to which it can be explained by covariates serving as explanatory factors in our regression model on ideological positions.

% Secondly, for. e.g. Topic~9 we mentioned above that the polarity is not driven by party assignment. The fixed ideological position model is not able to capture this, resulting in a very low polarity score. However the more flexible model allows to detect this.  
\begin{figure}[t!]
    \centering
     \includegraphics[width=\textwidth]{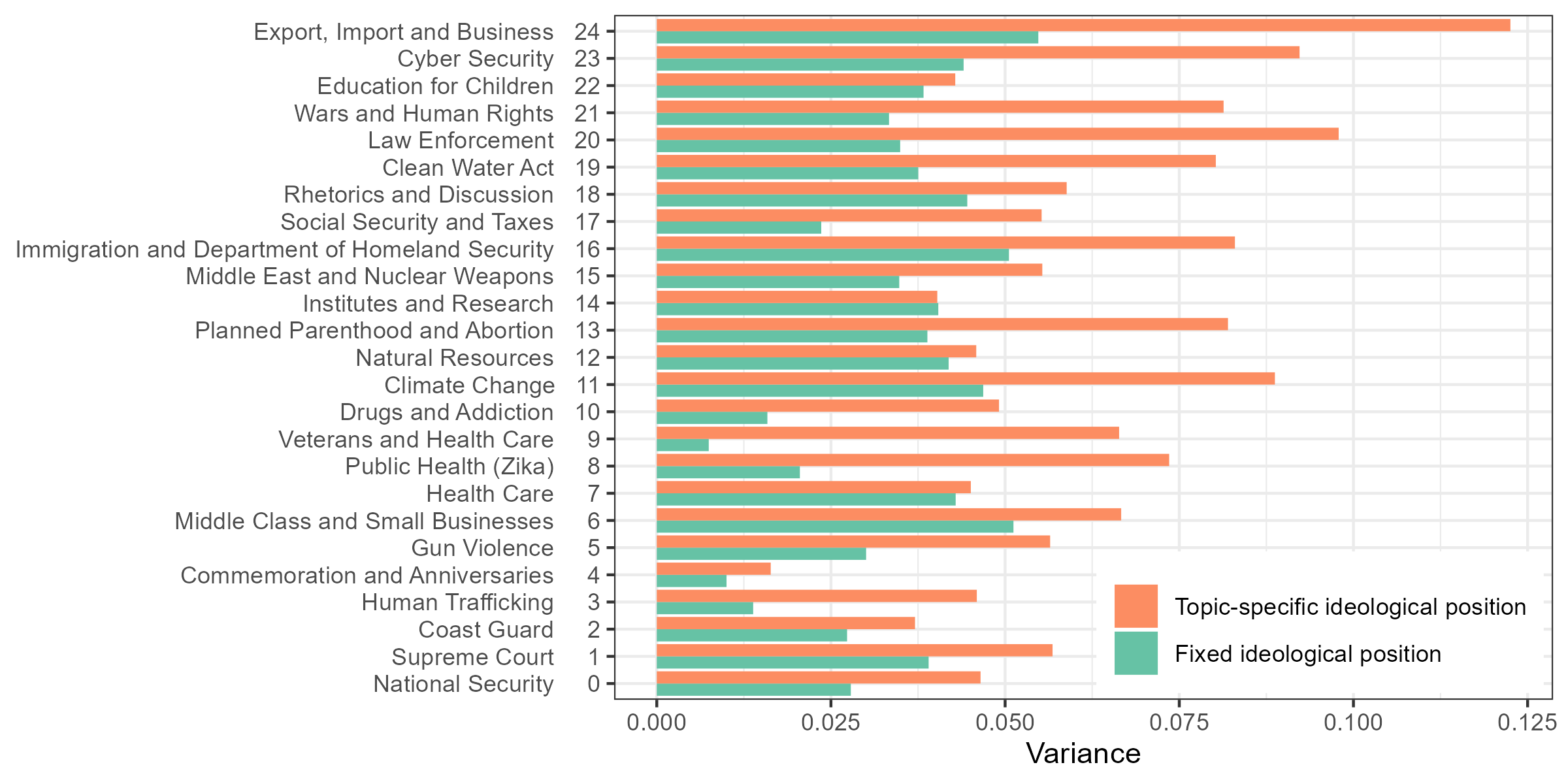}
    \caption{Comparison of polarity induced by the variance of $\phi_{\eta kv}^\loc \cdot \phi_{\ideal ak}^\loc$ under fixed and under topic-specific ideological positions.}
    \label{fig:eta_ideal_variability_a_vs_ak}
\end{figure}

\subsection[Fixed ideological positions]{Fixed ideological positions $\ideal_a$} \label{subsec:reg_ideal_a}
% In this section, we present the results of a simple specification of our proposed STBS model where the ideological positions $\ideal_a$ are assumed to be identical over topics for each author but may be driven through author specific covariates.
% In particular, we use an additive regression set up including the following categorical covariates (baseline category): party (Democrats), gender (male), region (Northeast), generation (Silent), experience (10+ years in Congress) and religion (Other), i.e.
% \begin{equation}
% \label{eq:example1}
% \texttt{ideal $\sim$ party + gender + region + generation + experience +  religion}
% \end{equation}
% This specification represents a direct extension of the TBIP model proposed in \citet{Vafa_etal_2020} by including the structural information of the text originator covariates.  
% In  what follows, we focus on the results of the scaling positions $\ideal_a$, topic word distributions will be a presented in the more general model setting of STBS. 

\begin{figure}[t!]
    \centering
    \includegraphics[width=\textwidth]{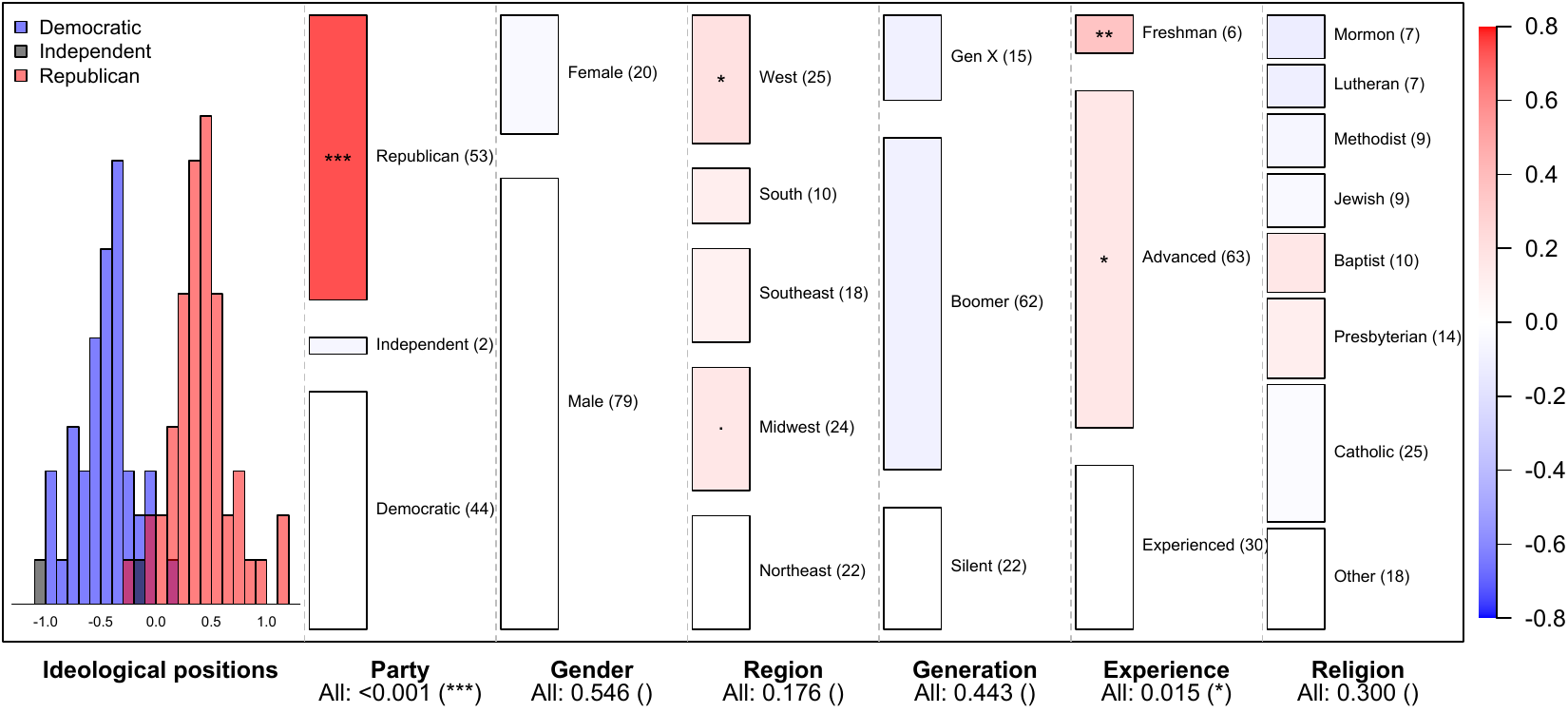}
    \caption{Regression~\eqref{eq:example1} summary plot of the model with fixed ideological positions across topics. CCP codes: $*** < 0.001$, $** < 0.01$, $* < 0.05$, $\cdot < 0.1$.}
    \label{fig:TBIPhier_ideal_a_Nreg_all_no_int_covariate_effects}
\end{figure}

The regression results of Equation~\eqref{eq:example1} are summarized in Figure~\ref{fig:TBIPhier_ideal_a_Nreg_all_no_int_covariate_effects}. The histogram on the left is grouped by party membership and clearly shows the strong dependence of the ideological positions on party membership. We obtain a bimodal distribution for the ideological positions $\ideal_a$ with in general negative values for the Democrats and positive values for the Republicans. The influence of the party affiliation is further supported by the low CCP value given at the bottom of the first column labeled with ``Party'' (CCP $<0.001$).

The other CCP values at the bottom of Figure~\ref{fig:TBIPhier_ideal_a_Nreg_all_no_int_covariate_effects} indicate that there is posterior evidence that years of service for the Congress influences the ideological position (CCP for ``Experience'' $= 0.015$).
The colors in the rectangles imply that the ideological position shifts to the right with a decrease in experience.
Keeping in mind that all Senators categorized as ``Freshman'' are Republican, this implies that in particular less experienced Republicans lean more towards using extreme language than their party colleagues. One possible reason, albeit speculative, is that less experienced (Republican) Senators may seek to gain more attention and solidify their position within their party by adopting more extreme stances, thereby asserting their political identity.

For this model, all other CCP values included at the bottom are above 0.05 and indicate a lack of posterior evidence that any other covariate influences consistently the ideological position of a speaker.

\subsection[Topic-specific ideological positions]{Topic-specific ideological positions $\ideal_{ak}$} \label{subsec:detailed_regression}

In the following, we consider in detail the results for the STBS model with topic-specific ideological positions $\ideal_{ak}$ using Equation~\eqref{eq:example2} in the regression setup. We first focus on the distribution of ideological positions grouped by party for each of the topics and averaged across topics. The average values are compared to those obtained with the TBIP model imposing fixed ideological positions. 

We inspect the topical content based on word clouds reflecting the term frequency assuming a neutral ideological position as well as extreme ideological positions with values $1$ and $-1$ for selected topics (see Figure~\ref{fig:ideal_points_hier_vs_TBIP_party_by_weighted_average_K25} on the right).
Inspecting these word clouds allows to infer the content of a topic but also to assess how polarization changes word use in a debate on a specific topic. These word clouds are complemented by those in Figure~\ref{fig:TBIPhier_all_wordclouds_justeta_logscale_selected_topics} in Appendix~\ref{app:wordclouds} where for negative and positive polarity focus is given to those terms where the increase is strongest according to the polarity values. To infer the topical content as well as how ideology changes word use, we also examined speeches with particularly high $\theta$ values for a specific topic to infer the topical content in combination with the topic-specific ideological position inferred for the Senator giving the speech. More details on the most relevant speeches identified in this way for a selected set of topics are given in Appendix~\ref{app:speeches}.

We focus in particular on the following topics. We inspect results for topic~4 (Commemoration and Anniversaries) and topic~24 (Export, Import and Business) which had the lowest and highest values for polarity. In addition we consider topic~9 (Veterans and Health Care), where the difference in polarity between the two model specifications was strongest. We also inspect topic~16 (Immigration and Department of Homeland Security) where the difference between the party averages for the topic-specific ideological positions is largest. Also topic~11 (Climate Change) and topic~13 (Planned Parenthood and Abortion) are inspected in detail. These topics have rather high polarity values as well as differences between party averages and the regression analysis indicates interesting interaction effects.  

\begin{sidewaysfigure}
  \centering  
\begin{minipage}{0.495\textwidth}
 \includegraphics[width=0.99\textwidth]{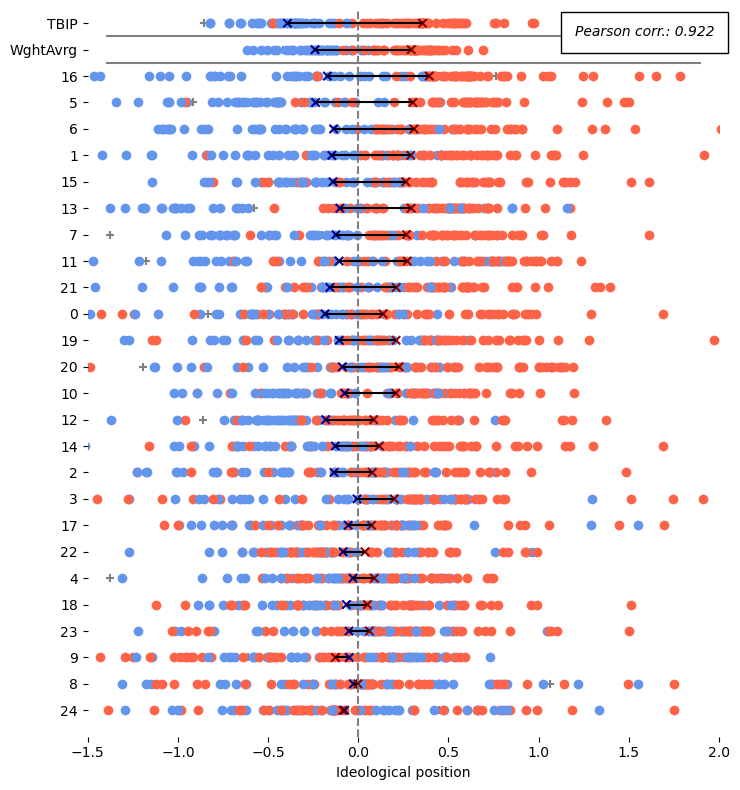}
    
\end{minipage}
\begin{minipage}{0.495\textwidth}
    \includegraphics[width=0.99\textwidth]{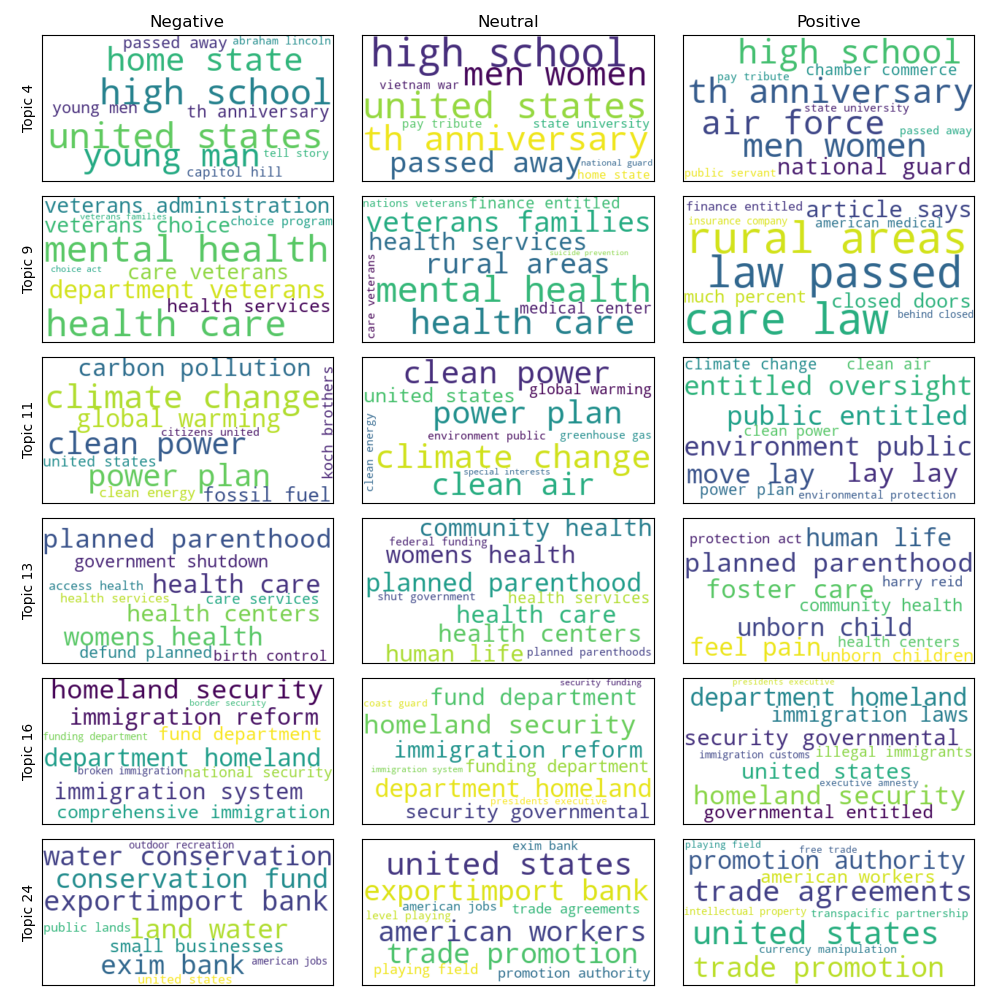}
    %\label{fig:TBIPhier_all_wordclouds_logscale_selected_topics}
\end{minipage}

\caption{Left: topic-specific ideological positions $\ideal_{ak}$; 
    Democrats (blue), Republicans (red), Independent (grey). 
    %Democrats (lightblue) are represented by Charles Schumer (blue) and Republicans (orange) are represented by Mitch McConnell (red). 
    Black lines connect the weighted averages of ideological positions for Democrats and Republicans and topics are ordered by the distance between the weighted averages.
    Right: word clouds for selected topics. Columns labeled ``Positive'' and ``Negative'' include results for the intensities corrected for ideology using an ideological position equal to $1$ or $-1$, respectively.}
\label{fig:ideal_points_hier_vs_TBIP_party_by_weighted_average_K25}
\end{sidewaysfigure}

Figure~\ref{fig:ideal_points_hier_vs_TBIP_party_by_weighted_average_K25} on the left presents the estimated ideological positions at the topic level for all $K=25$ topics
with Democratic Senators in blue and Republican Senators in red. The topics are ordered such that the distance between the average ideological positions of Democrats and Republicans is decreasing. These averages are indicated by crosses and linked by a black line. Topics at the top are the most polarizing ones between the parties. 

In addition the ideological positions obtained for the TBIP model are included as well as the weighted average ideological positions across topics. The Pearson's correlation coefficient between these weighted average ideological positions and those obtained for TBIP exceeds 0.9, thus signifying a strong congruence and suggesting that the overall ideological stance is maintained by the STBS model while enabling a selective diversification of ideological positions across specific topics.

\begin{figure}[t!]
    \centering
    \includegraphics[width=0.99\textwidth]{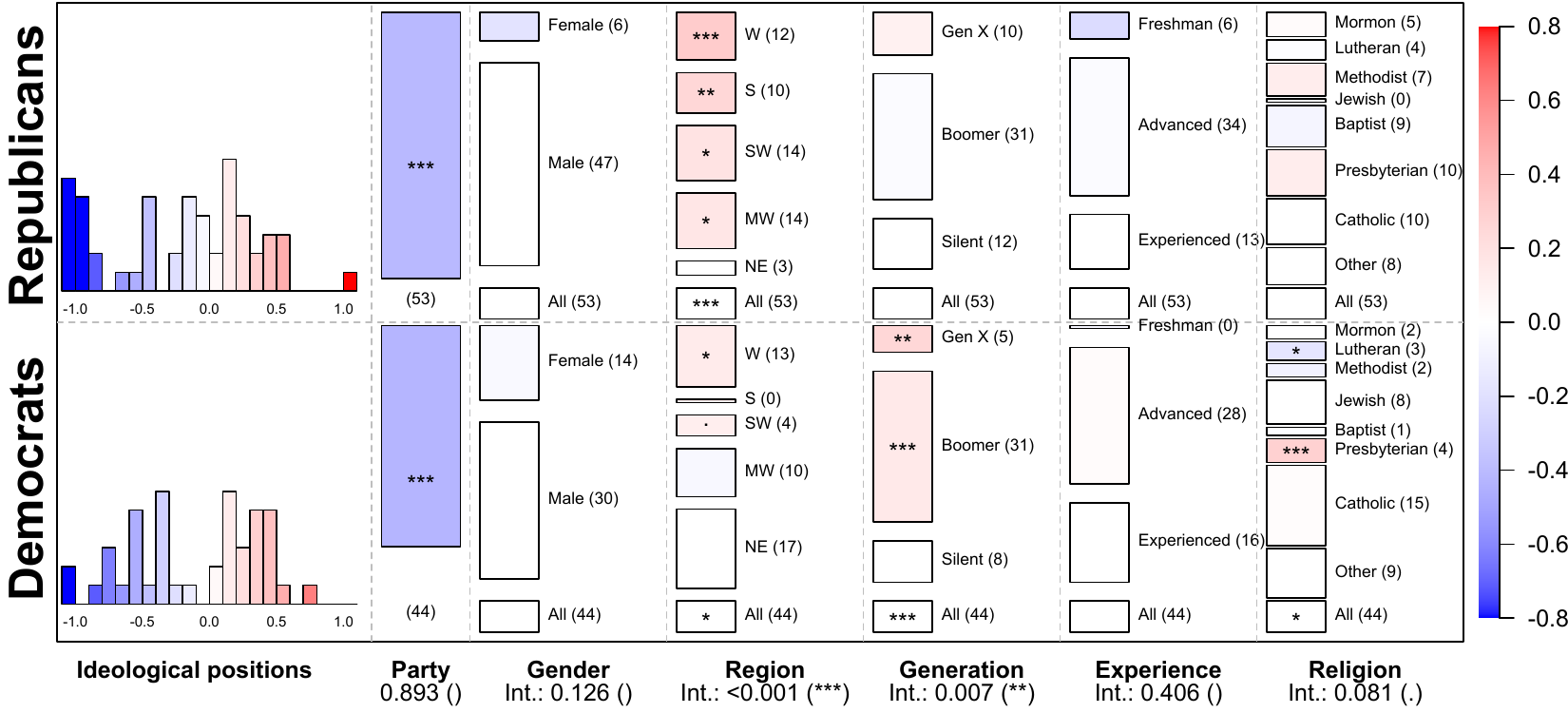}
    \caption{Topic~9 (Veterans and Health Care). Regression~\eqref{eq:example2} summary plot for the model with topic-specific ideological positions.
    CCP codes: $*** < 0.001$, $** < 0.01$, $* < 0.05$, $\cdot < 0.1$.}
    \label{fig:TBIPhier_all_party_effects_interactions_transposed_k_9}
\end{figure}

Figure~\ref{fig:ideal_points_hier_vs_TBIP_party_by_weighted_average_K25} (left) indicates that topic~9 (Veterans and Health Care) is the only one topic where the order of the party means is switched. However, their difference is only very small. This suggests that the polarity of this topic is not determined by party assignment. Analyzing the regression results visualized in Figure~\ref{fig:TBIPhier_all_party_effects_interactions_transposed_k_9} provides insights into the drivers of the ideological positions of this topic.  Given the speaker is a Republican and all other covariates take the values of the baseline categories (i.e., Male, NE, Silent, Experienced, Other), the ideological position is shifted to the left (rectangle at the top left). The same holds for Democrats (rectangle at the bottom left). For both parties there is strong posterior evidence that region influences the scaling position with more positive ideological positions for Senators from the West, Southwest and South. Inspecting the corresponding word clouds and speeches reveals that speakers from those regions tend to talk more about the medical and health care system, and not about veterans. One rather speculative reason for this could be that the health care system in general has a strong geographic variation, while the veterans' health care system is centralized \citep[see][]{Ashton}. % \textcolor{red}{speeches of topic 9?}

% This geographic variation is also observed for topic~5 (on gun control), which is the most polarizing topic in ideological positions. Here we observe strong posterior evidence of party having an influence, with a coefficient of $0.784$ for Republican. For the interaction term between Republican and region, we observe the strongest effect for \emph{Midwest} with a coefficient of $0.249$ followed by \emph{West} ($0.227$) and \emph{Southeast} ($0.101$) (see the results in Appendix~\ref{sec:emp-analysis-details}).

%We also mentioned that the highest polarization in language is observed for topic~24. Following the word clouds this is a topic about \emph{}
\begin{figure}[t!]
    \centering
    \includegraphics[width=0.99\textwidth]{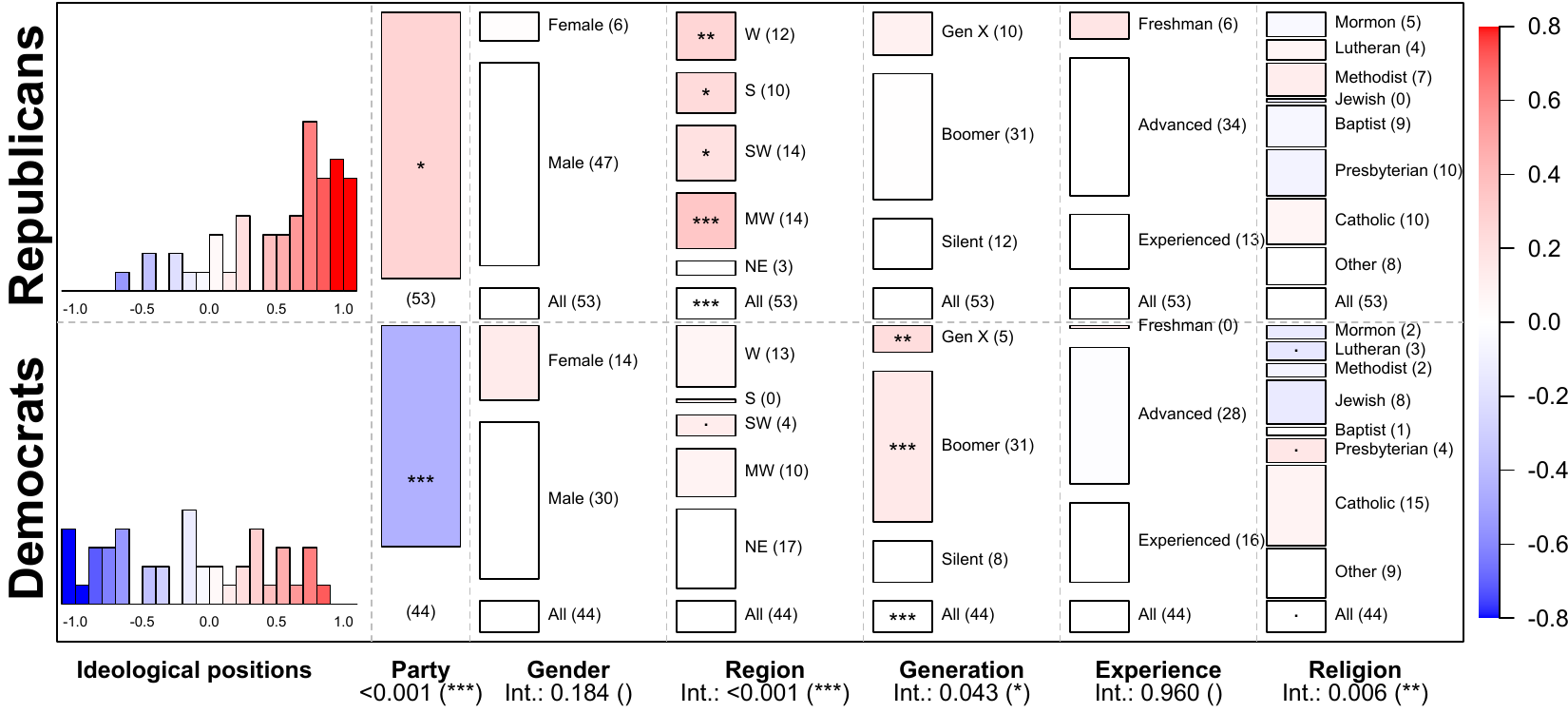}
    \caption{Topic~11 (Climate Change). Regression~\eqref{eq:example2} summary plot for the model with topic-specific ideological positions.
    CCP codes: $*** < 0.001$, $** < 0.01$, $* < 0.05$, $\cdot < 0.1$.}
    \label{fig:TBIPhier_all_party_effects_interactions_transposed_k_11}
\end{figure}

Figure~\ref{fig:TBIPhier_all_party_effects_interactions_transposed_k_11} summarizes the regression results for topic~11 (Climate Change). For this topic, we do not only observe ideological discrepancies between the parties but also between the generations, in particular for the Democrats (CCP value   $<0.001$ in Table~\ref{tab:regression_coefs_k_4_9_11_13_16_24} in Appendix~\ref{app:coefs}). Ideological positions of the younger generation tend to be more positive for Democrats. Positive extremists would rather talk about entitled oversight while negative extremists frequently use \emph{fossil fuel} or \emph{carbon pollution}. 

\begin{figure}[t!]
    \centering
    \includegraphics[width=0.99\textwidth]{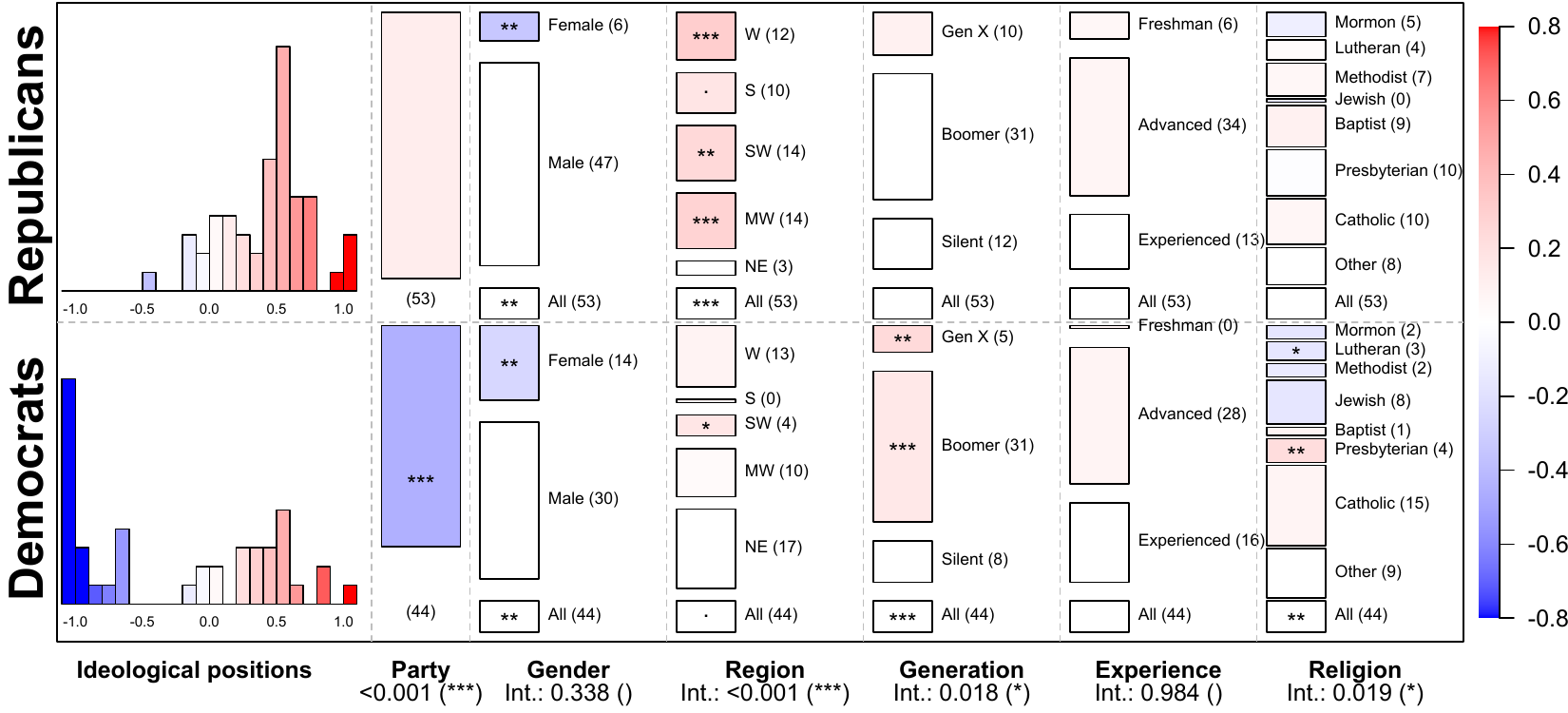}
    \caption{Topic~13 (Planned Parenthood and Abortion). Regression~\eqref{eq:example2} summary plot for the model with topic-specific ideological positions.
    CCP codes: $*** < 0.001$, $** < 0.01$, $* < 0.05$, $\cdot < 0.1$.}
    \label{fig:TBIPhier_all_party_effects_interactions_k_13}
\end{figure}

Topic~13 (Planned Parenthood and Abortion) illustrates a topic with a significant effect for gender, see Figure~\ref{fig:TBIPhier_all_party_effects_interactions_k_13}. Regardless of party membership, women tend to have a more negative ideological position than men. From the word clouds in Figure~\ref{fig:ideal_points_hier_vs_TBIP_party_by_weighted_average_K25}, we see that a speaker with negative ideological position prefers to use terms like \emph{women's health}, \emph{family planning} or \emph{care services}, while a speaker with positive ideological position would talk about \emph{unborn children} and \emph{foster care}.

According to Figure~\ref{fig:eta_ideal_variability_a_vs_ak} inspecting induced polarity, topic~24 (Export, Import and Business) is overall the most polarizing topic.  This topic is at the bottom of Figure~\ref{fig:ideal_points_hier_vs_TBIP_party_by_weighted_average_K25} thus indicating that this polarization is not driven by party membership. % Inspecting the word clouds and documents assigned to that topic identifies topic~24 as an economic topic focusing on trade and American workers where the included bigrams are \emph{import export, trade agreements, american workers}.
Inspecting the regression results (see 
Table~\ref{tab:regression_coefs_k_4_9_11_13_16_24} in Appendix~\ref{app:coefs}) reveals that indeed generation in combination with party as well as region and religion for Republicans is associated with the ideological position.

%% file: sec05_conclusion.tex
\section{Conclusion} \label{sec:conclusion}

Previous approaches to infer ideological positions from texts were limited in several aspects. The STBS model relaxes the assumption of fixed ideological positions across topics and allows ideological positions to vary across topics of discussion. 
Hierarchical shrinkage priors complement the generative model and increase flexibility as well as guide inference inducing sparsity.
A regression model that identifies the associations between the ideological positions and covariates characterizing the authors is included in the model and no additional post-hoc analysis is required to infer how covariates are associated with the ideological positions but their effects are directly captured by model parameters.  To estimate STBS, we combined the pure stochastic gradient approach pursued in \citet{Vafa_etal_2020} to find the optimal variational parameters with CAVI updates to improve the convergence behavior and save computational resources by reducing the number of gradients to track. These updates are available due to the conjugate choices for the hierarchical prior distributions and the variational families. 

The STBS model provides enhanced flexibility to the TBIP model by allowing for topic-specific ideological positions as well as the inclusion of covariates to characterize these positions. Nevertheless, the STBS model still relies on a number of assumptions regarding the ideological positions which might be perceived as rather restrictive in certain application cases. For example, the STBS model assumes that ideological positions remain constant for an author across all documents given topic. This is a reasonable assumption for a corpus where documents are collected in a rather short time frame. For corpora spanning a long time frame a time-varying version might be preferable. 
% E.g., if several sessions of the U.S.\ Senate are analyzed, one might want allow for session-to-session changes in ideological positions. Such an approach was already considered by \citet{Hofmarcher+Adhikari+Gruen:2022} for the TBIP model. In principle, our STBS model could be fitted in a similar way. Alternatively, the document-term matrices from different sessions could be combined into one large corpus and analyzed with the STBS model where a combination of author and session is used as author index. The regression on ideological positions could then be enriched by session-specific covariates. % Nevertheless, this solution is still restricted by assuming that ideological positions are fixed throughout the whole session and they are only allowed to change when a new session starts. We consider the assumption of changing the ideological position after a~very influential real-world event far more realistic. Detection of these events requires strict control over the prior distribution for document-specific ideological positions which is the next challenge we would like to address. 

Another restrictive aspect of the STBS model is that for a specific topic only two opposing ideologies are allowed to exist. In the world of bipolar politics as in the U.S., this  represents a reasonable assumption. However, in European countries the political situation is more complex and a one-dimensional concept of ideology would seem insufficient to capture all the nuances. Inspired by the Political Compass%\citep{wiki:The_Political_Compass}
, we could extend the STBS model and construct a text-based model that distinguishes the ideology on a given topic using both the social scale (Libertarian/Authoritarian) and the economic scale (Left/Right). % However, this requires a proper definition and implementation of the latent space so that the unique results are obtained despite the invariance to rotations.  

\begin{comment}
\begin{enumerate}
    \item Summarize our upgrades for TBIP model and how it improved the real data analysis.
    \item Suggest possible improvements:
    \begin{enumerate}
        \item Include also ethnicity of the author.
        \item Incorporate document-specific information on top of author-level information. This could model the evolution in time directly. But only document-specific parameters could arise from such regression model, e.g. $\theta_{dk}$, which was already used in STM \citep{Roberts_etal_JASA:2016}. Or somehow make the positions document-specific, e.g. a curve in time instead of a~fixed position. 
        \item Bi(Multi)variate ideological positions, there might not be only one direction to explore. Useful for European situation with many political parties. But problem with the fixation of the latent ideological space. 
        \item Other real datasets to analyze.
    \end{enumerate}
\end{enumerate}
\end{comment}

%% file: appendix.tex
\section{Implementation details} \label{sec:implementation-details}

%%%!!!!!!!!!!!!!!!!!!!!!!!!!!!!!!!!!!!!!!!!!!!!!!!!!!!!!!!!!!!!!!!!!!!!!!!%%%
%%%%%%%%%%%%%%%%%%%%%%%%%%%%%%%%%%%%%%%%%%%%%%%%%%%%%%%%%%%%%%%%%%%%%%%%%%%%%
%%%%%%%%%%%%               NEW SUBSECTION                   %%%%%%%%%%%%%%%%%
%%%%%%%%%%%%%%%%%%%%%%%%%%%%%%%%%%%%%%%%%%%%%%%%%%%%%%%%%%%%%%%%%%%%%%%%%%%%%
%%%!!!!!!!!!!!!!!!!!!!!!!!!!!!!!!!!!!!!!!!!!!!!!!!!!!!!!!!!!!!!!!!!!!!!!!!%%%
\subsection[Expected ideological term]{Expected ideological term}\label{app:expectedterm}

The STBS model differs from the classical Poisson factorization model primarily because of the ideological term $\exp\{\eta_{kv} \ideal_{ak}\}$ that captures the different ideologies. This term includes the exponentiated product of two normally distributed variables which breaks conjugacy and complicates estimation. 
The CAVI updates require the expected ideological term $E_{dkv}$.  \citet{Vafa_etal_2020} use the following approximation:
$$
E_{dkv} := \Eq \exp\{\eta_{kv} \ideal_{a_d k}\} \approx \exp \{ \Eq \eta_{kv} \Eq \ideal_{a_d k} \} = \exp \{\phi_{\eta kv}^\loc \phi_{\ideal a_d k}^\loc \},
$$
claiming that the expectation is intractable. Thus, they replace $E_{dkv}$ with this geometric mean. 
However, the expectation can be obtained as follows.

\begin{alemma} \label{lemma:E_exp_XY}
    Let $X \sim \norm{\mu_x}{\sigma_x^2}$ and $Y \sim \norm{\mu_y}{\sigma_y^2}$ be independent. Then,
    \begin{equation}
        \E \exp \{X Y\} = \dfrac{1}{\sqrt{1-\sigma_{x}^2\sigma_{y}^2}} \exp 
        \left\{
            \dfrac{1}{2} \dfrac{\mu_{x}^2\sigma_{y}^2 + 2\mu_{x}\mu_{y} + \mu_{y}^2\sigma_{x}^2}{1-\sigma_{x}^2\sigma_{y}^2}
        \right\} 
        \quad \text{only if}\;
        \sigma_{x}^2\sigma_{y}^2 < 1,
    \end{equation}
    otherwise it is $+\infty$.
\end{alemma}

\begin{proof}
    Use the rules for conditional expectation with the following intermediate results:
    $$
    \E e^{tX} = \exp \left\{\mu_x t + \frac{1}{2} \sigma_x^2 t^2 \right\} 
    \quad \text{and}\;
    \E e^{tX^2} = \dfrac{1}{\sqrt{1-2t\sigma_x^2}} \exp \left\{ \dfrac{\mu_x^2 t}{1-2t\sigma_x^2}\right\}
    \; \text{if}\; 
    2t\sigma_x^2 < 1.
    $$
\end{proof}

Lemma~\ref{lemma:E_exp_XY} indicates that the geometric approximation only provides  reasonable results if the variances are close to zero. % However, in their original implementation no restrictions are imposed on these variances and hence, they could potentially explode for a poor initialization. 
To ensure a finite expectation, we impose the restrictions $\bm \phi_{\ideal}^\sclsq < 1$ and $\bm \phi_{\eta}^\sclsq < 1$. 
We enforce this restriction by using the \emph{sigmoid} bijector instead of the \emph{softplus} bijector in the Tensorflow environment. 

The implementation allows the user to choose between using the original, rather straightforward geometric approximation or the exact, but more complicated expectation given by:
\begin{equation} \label{eq:Edkv}
    E_{dkv} 
	=
	\dfrac{1}{\sqrt{1-\phi_{\eta kv}^\sclsq \phi_{\ideal a_d k}^\sclsq}} \exp 
	\left\{
	\dfrac{1}{2} \dfrac{\left(\phi_{\eta kv}^\loc\right)^2 \phi_{\ideal a_d k}^\sclsq + 2\phi_{\eta kv}^\loc \phi_{\ideal a_d k}^\loc + \left(\phi_{\ideal a_d k}^\loc\right)^2 \phi_{\eta kv}^\sclsq}{1-\phi_{\eta kv}^\sclsq \phi_{\ideal a_d k}^\sclsq}
	\right\}.
\end{equation}
We use the exact expectation in the empirical applications presented. 
%Using the exact expectation, CAVI updates for the parameters included in \eqref{eq:Edkv} are not available. Hence, the stochastic gradient approach is used instead to update them.
%\textcolor{red}{Not true, parameters included in \eqref{eq:Edkv} have to be updated with stochastic gradient approach even with the geometric approximation.}

%%%!!!!!!!!!!!!!!!!!!!!!!!!!!!!!!!!!!!!!!!!!!!!!!!!!!!!!!!!!!!!!!!!!!!!!!!%%%
%%%%%%%%%%%%%%%%%%%%%%%%%%%%%%%%%%%%%%%%%%%%%%%%%%%%%%%%%%%%%%%%%%%%%%%%%%%%%
%%%%%%%%%%%%               NEW SUBSECTION                   %%%%%%%%%%%%%%%%%
%%%%%%%%%%%%%%%%%%%%%%%%%%%%%%%%%%%%%%%%%%%%%%%%%%%%%%%%%%%%%%%%%%%%%%%%%%%%%
%%%!!!!!!!!!!!!!!!!!!!!!!!!!!!!!!!!!!!!!!!!!!!!!!!!!!!!!!!!!!!!!!!!!!!!!!!%%%
\subsection{Details on the combination of SVI and ADVI}\label{app:stochastic-gradients}

\citet{Vafa_etal_2020} estimate the TBIP model solely based on a variant of ADVI where they determine the step size based on the Adam optimizer~\citep{kingma_ba2015Adam}. Although, the evaluation of the gradients is relatively cheap, tracking them for thousands of parameters is not. Therefore, we pursue a similar implementation as also used for the latest formulation of TBIP on Github \citep{Github_TBIP}. This implementation brings the best of both worlds together: CAVI updates are employed where available and automatic differentiation is employed to perform the update steps otherwise. 

To optimize $\ELBO{\bm\phi}$, some of the parameters are moved in the direction of the steepest ascent while other parameters are kept fixed. The contributions to $\ELBO{\bm \phi}$ are essentially expectations of log-densities with respect to the variational families ($\Eq$). By the Monte Carlo principle, this expectation can be approximated by averaging over log-densities evaluated at samples from the variational families. In ADVI, the gradients are tracked by the reparametrization trick \citep{KingmaWelling2014}. However, this trick is unavailable for gamma variational families. This is the reason why \citet{Vafa_etal_2020} used the log-normal variational family for positive parameters in order to be able to implement a pure ADVI procedure.

To obtain the Monte Carlo approximation, one could proceed as follows: One draws samples $\bm \zeta^i, i = 1, \ldots, I$ from the variational family $q_{\bm \phi^t}$ with current parameter values $\bm \phi^t$ and approximates the current value of the $\mathsf{ELBO}$ by
$$
\ELBO{\bm \phi^t} \approx \dfrac{1}{I}\sum\limits_{i=1}^I \left[\log p(\mathbb{Y} | \bm \zeta^i) + \log p(\bm \zeta^i | \mathbb{X}) - \log q_{\bm \phi^t} (\bm \zeta^i)\right].
$$
This is straightforward to implement in Tensorflow \citep[see also][]{Vafa_etal_2020}. Tensorflow also automatically tracks the gradients using automatic differentiation including the reparameterization trick for most of the distributions (except for the gamma distribution). The reparameterization trick implies that, for example, the random variable $\ideal_{ak}^i$ distributed as $\norm{\phi_{\ideal ak}^\loc}{\phi_{\ideal ak}^\sclsq}$ is reparameterized as $\ideal_{ak}^i = \phi_{\ideal ak}^\scl z_{\ideal ak}^i + \phi_{\ideal ak}^\loc$ with auxiliary $z_{\ideal ak}^i \sim \norm{0}{1}$, from which the gradients can be easily derived with $\bm z^i$ as constants. 

This straightforward approach employs the Monte Carlo approximation for all parts of the $\mathsf{ELBO}$. However, examining the individual parts of the $\mathsf{ELBO}$ function indicates that for some parts the approximate solution could be avoided. For example, the last term $-\Eq \log q_{\bm \phi}(\bm\zeta)$ corresponds to the \emph{entropy} of the variational distribution, for which an explicit expression exists \citep[see also][]{Kucukelbir+Tran+Ranganath:2017}. It clearly is more accurate and less time-consuming to evaluate the entropy using the explicit expression once than averaging $I$ log-densities evaluated at random points and then track the gradients. Our implementation makes use of this exact \emph{entropy} calculation.

\subsection{CAVI updates}\label{app:cavi}
In this section we list all CAVI updates used in the implementation. The updates will be denoted by $\widehat{\bm \phi}_\bullet$.
The only parameters for which CAVI updates are not available are the variational parameters $\bm \phi_\ideal$ and $\bm \phi_{\eta}$ for $\bm \ideal$ and $\bm \eta$. This is due to the non-conjugate form of the ideological term. Some of the updates are inspired by~\citet{gopalanhofmanblei2015scalableHPoisF}.
\citet{BleiVI:2017} provide a general approach to obtain CAVI updates. This approach exploits that the optimal variational distribution for a~block of parameters is proportional to the exponentiated expected log of the complete conditional distribution of the respective block. Alternatively, there is always an option to express $\ELBO{\bm \phi}$ as a~function of the current block of parameters up to an additive constant and find the optimal update by solving the system of estimating equations constructed by zeroing the gradients. 

% The fact that the Poisson rates $\lambda_{dv}$ are the sum of $\lambda_{dkv}$ over the topic dimension complicates direct CAVI update derivation. We resolve this by introducing auxiliary Poisson counts $\widetilde{\mathbb{Y}}$. The updates for the auxiliary proportions which divide $y_{dv}$ into $y_{dv} = \widetilde{y}_{d1v} + \cdots + \widetilde{y}_{dKv}$ are of the following form:
% $$
% \widehat{\phi}_{dkv}^y \propto
% \exp \left\{
% \psi(\phi_{\theta dk}^\shp) - \log\phi_{\theta dk}^\rte
% + \psi(\phi_{\beta kv}^\shp) - \log\phi_{\beta kv}^\rte
% + \phi_{\eta kv}^\loc \phi_{\ideal a_d k}^\loc
% \right\}.
% $$

We start with the parameters for $\bm \theta$, $\bm \beta$ and their hyperparameters: 
% We clearly see the similarity to updates by~\citet{gopalanhofmanblei2015scalableHPoisF}:
\begin{align}
    \widehat{\phi}_{\theta dk}^\shp &= a_\theta + %\sum\limits_{d\in \mathcal{D}_a}
	\sum\limits_{v=1}^V y_{dv} \phi_{dkv}^y
	&\text{and} \qquad
	\widehat{\phi}_{\theta dk}^\rte &= \dfrac{\phi_{b \theta a_d}^\shp}{\phi_{b \theta a_d}^\rte}
	+ \sum\limits_{v=1}^V \dfrac{\phi_{\beta kv}^\shp}{\phi_{\beta kv}^\rte} E_{dkv},
 \label{eq:cavi-theta}
    \\
	\widehat{\phi}_{b \theta a}^\shp &= a_\theta' + K |\mathcal{D}_a | a_\theta 
	&\text{and} \qquad
	\widehat{\phi}_{b \theta a}^\rte &= \dfrac{a_\theta'}{b_\theta'}
	+ \sum\limits_{d\in \mathcal{D}_a} \sum\limits_{k=1}^K \dfrac{\phi_{\theta dk}^\shp}{\phi_{\theta dk}^\rte},
 \label{eq:cavi-btheta}
    \\
    \widehat{\phi}_{\beta kv}^\shp &= a_\beta + \sum\limits_{d=1}^D y_{dv} \phi_{dkv}^y
	&\text{and} \qquad
	\widehat{\phi}_{\beta kv}^\rte &= \dfrac{\phi_{b \beta v}^\shp}{\phi_{b \beta v}^\rte}
	+ \sum\limits_{d=1}^D \dfrac{\phi_{\theta dk}^\shp}{\phi_{\theta dk}^\rte} E_{dkv},
 \label{eq:cavi-beta}
	\\
	\widehat{\phi}_{b \beta v}^\shp &= a_\beta' + K a_\beta
	&\text{and} \qquad
	\widehat{\phi}_{b \beta v}^\rte &= \dfrac{a_\beta'}{b_\beta'}
	+ \sum\limits_{k=1}^K \dfrac{\phi_{\beta kv}^\shp}{\phi_{\beta kv}^\rte}.
 \label{eq:cavi-bbeta}
\end{align}
The updates in \eqref{eq:cavi-theta} are only performed for the documents in the current batch, $d \in \mathcal{B}_b$, see Algorithm~\ref{alg:TBIP_CAVI_and_SVI}. For the updates in \eqref{eq:cavi-btheta}, the documents $\mathcal{D}_a$ written by author~$a$ have to be restricted to include only those from the current batch and are accordingly rescaled. The updates in \eqref{eq:cavi-beta} for the $\bm \beta$ parameters have to be restricted and rescaled in the same way to take only documents in the current batch into account. This yields the following stochastic, batch-specific CAVI updates:
\begin{align*}
	\widehat{\phi}_{b \theta a}^\shp &= a_\theta' + \dfrac{D}{|\mathcal{B}_b|} K |\mathcal{D}_a \cap \mathcal{B}_b | a_\theta 
	&\text{and} \qquad
	\widehat{\phi}_{b \theta a}^\rte &= \dfrac{a_\theta'}{b_\theta'}
	+ \dfrac{D}{|\mathcal{B}_b|} \sum\limits_{d\in \mathcal{D}_a \cap \mathcal{B}_b} \sum\limits_{k=1}^K \dfrac{\phi_{\theta dk}^\shp}{\phi_{\theta dk}^\rte},
    \\
    \widehat{\phi}_{\beta kv}^\shp &= a_\beta + \dfrac{D}{|\mathcal{B}_b|} \sum\limits_{d \in \mathcal{B}_b} y_{dv} \phi_{dkv}^y
	&\text{and} \qquad
	\widehat{\phi}_{\beta kv}^\rte &= \dfrac{\phi_{b \beta v}^\shp}{\phi_{b \beta v}^\rte}
	+ \dfrac{D}{|\mathcal{B}_b|} \sum\limits_{d \in \mathcal{B}_b} \dfrac{\phi_{\theta dk}^\shp}{\phi_{\theta dk}^\rte} E_{dkv}.
\end{align*}
Randomness is induced by the batch selection, i.e., the splitting of the documents into batches. Finally, the updates in \eqref{eq:cavi-bbeta} are the first non-stochastic CAVI updates. 

We continue with the parameters that form the triple-gamma prior for $\bm\eta$, namely $\bm \rho^2$ and $\bm b_\rho$, see~\eqref{eq:prior_eta}. This triple-gamma specification \citep{Knaus:2023dynamic} allows for direct CAVI updates due to the normal-gamma conjugacy in contrast to the original specification~\eqref{eq:prior_triple_gamma}. The updates for the shape parameters are fixed:
$$
\widehat{\phi}_{\rho k}^\shp = a_\rho + \frac{V}{2}
\qquad \text{and} \qquad
\widehat{\phi}_{b\rho k}^\shp = a_\rho' + a_\rho.
$$
The rates, on the other hand, require an update in every iteration as they depend on other variational parameters:
$$
\widehat{\phi}_{\rho k}^\rte = \frac{\phi_{b\rho k}^\shp}{\phi_{b\rho k}^\rte} + \frac{1}{2} \sum\limits_{v=1}^V \left[\left(\phi_{\eta kv}^\loc\right)^2 + \phi_{\eta kv}^\sclsq\right]
\qquad \text{and} \qquad 
\widehat{\phi}_{b\rho k}^\rte = \frac{\kappa_\rho^2}{2} \frac{a_\rho'}{a_\rho} + \dfrac{\phi_{\rho k}^\shp}{\phi_{\rho k}^\rte}.
$$

Next, we provide the CAVI updates for the parameters of the author-level regression of the ideological positions $\ideal$ in \eqref{eq:prior_ideal}. The updates for the author-specific precisions of the error terms are straightforward:
$$
\widehat{\phi}_{Ia}^\shp = a_I + \frac{K}{2}
\qquad \text{and} \qquad
\widehat{\phi}_{Ia}^\rte = b_I + \frac{1}{2} \sum\limits_{k=1}^{K}
\left[
\left( \phi_{\ideal ak}^\loc - \T{\bm x_a} \bm \phi_{\iota k}^\loc \right)^2
+ \phi_{\ideal ak}^\sclsq
+ \sum\limits_{l=1}^L x_{al}^2 \phi_{\iota kl}^\sclsq
\right].
$$
The regression coefficients $\bm \iota$ given in~\eqref{eq:prior_iota} are updated by combining the prior location and the classical least squares estimator for linear regression. If we denote by $\bm E_\omega$ the vector of all variational means of $\omega_l^2$, i.e., $E_{\omega l} = \frac{\phi_{\omega l}^\shp}{\phi_{\omega l}^\rte}$, then we can write the updates in the following form:
\begin{align*}
\widehat{\bm\phi}_{\iota k}^\loc 
&=
\left[
\diag \bm E_\omega + \sum\limits_{a=1}^A \dfrac{\phi_{Ia}^\shp}{\phi_{Ia}^\rte}   \bm x_a \T{\bm x_a}
\right]^{-1}
\left(
\left[\diag \bm E_\omega \right] \bm \phi_{\iota \bullet}^\loc
+
\sum\limits_{a=1}^A \dfrac{\phi_{Ia}^\shp}{\phi_{Ia}^\rte} 
\phi_{\ideal ak}^\loc \bm x_a
\right),
\\
\widehat{\phi}_{\iota kl}^\sclsq
&= 
\dfrac{1}{E_{\omega l} + \sum\limits_{a=1}^A \dfrac{\phi_{Ia}^\shp}{\phi_{Ia}^\rte} x_{al}^2}, 
\end{align*}
if the variational family consists of independent normal distributions. Using a multivariate normal variational family for the coefficients, we get the following update for the covariance matrix:
$$
\widehat{\bm \phi}_{\iota k}^\covm = 
\left[
    \diag \bm E_\omega + \sum\limits_{a=1}^A \dfrac{\phi_{Ia}^\shp}{\phi_{Ia}^\rte}   \bm x_a \T{\bm x_a}
\right]^{-1}.
$$
Note that this matrix does not depend on $k$ since the model matrix is the same for all $K$ regressions. Therefore, only one matrix $\bm \phi_{\iota 1}^\covm$ is needed. In practice, the covariance matrix is parameterized through the lower triangular Cholesky factor. 
%In practise, tensorflow allows to parameterize the MV normal distribution through lower triangular Cholesky factor instead of the covariance. Then, the updates are based on the Cholesky decomposition of inverted matrix which is efficiently found through Cholesky decomposition of the original matrix. Traditional 'Cholesky solve' procedure is used to obtain the updates for locations.  

The variational parameters for the non-centrality parameters $\iota_{\bullet l}$ are updated in the following way:
$$
%\phi_{\iota\bullet l}^\loc = \phi_{\iota\bullet l}^\sclsq \left( \dfrac{1}{\sigma_{\iota\bullet l}} \mu_{\iota\bullet l} +  E_{\omega l}\sum\limits_{k=1}^{K} \phi_{\iota kl}^\loc \right)
%\qquad \text{and} \qquad
%\phi_{\iota\bullet l}^\sclsq = \dfrac{1}{\dfrac{1}{\sigma_{\iota\bullet l}} + K E_{\omega l}}.
\widehat{\phi}_{\iota\bullet l}^\loc = \widehat{\phi}_{\iota\bullet l}^\sclsq E_{\omega l}\sum\limits_{k=1}^{K} \phi_{\iota kl}^\loc 
\qquad \text{and} \qquad
\widehat{\phi}_{\iota\bullet l}^\sclsq = \dfrac{1}{1 + K E_{\omega l}}.
$$
The structure of the model implies that even under a multivariate normal variational family we obtain CAVI updates for $\bm \phi_{\iota\bullet}^\covm$ of diagonal form. Hence, for this parameter it is not necessary to use a multivariate normal variational distribution and we work only with $\phi_{\iota\bullet}^\sclsq$.

Finally, we have to update the parameters of the triple-gamma prior for $\bm \iota$, see~\eqref{eq:prior_omega}. Updates for $\bm \omega^2$ and $\bm b_\omega$ are analogical to the updates of $\bm \rho^2$ and $\bm b_\rho$. Parameters for $\bm \omega^2$ are updated as follows:
$$
\widehat{\phi}_{\omega l}^\shp = a_\omega + \frac{K}{2}
\qquad \text{and} \qquad 
\widehat{\phi}_{\omega l}^\rte = b_\omega + \frac{1}{2} \sum\limits_{k=1}^{K} \left[\left(\phi_{\iota kl}^\loc - \phi_{\iota\bullet l}^\loc\right)^2 + \phi_{\iota kl}^\sclsq + \phi_{\iota\bullet l}^\sclsq\right].
$$ 
The updates for the rates $b_{\omega l}$ are also straightforward:
$$
\widehat{\phi}_{b\omega l}^\shp = a_\omega' + a_\omega
\qquad \text{and} \qquad 
\widehat{\phi}_{b\omega l}^\rte = \frac{\kappa_\omega}{2} \frac{a_\omega'}{a_\omega} + \dfrac{\phi_{\omega l}^\shp}{\phi_{\omega l}^\rte}.
$$
Note that the shapes are constant and, therefore, are only determined once at the start of the algorithm.

\section{Additional empirical results}\label{sec:emp-analysis-details}
%%%!!!!!!!!!!!!!!!!!!!!!!!!!!!!!!!!!!!!!!!!!!!!!!!!!!!!!!!!!!!!!!!!!!!!!!!%%%
%%%%%%%%%%%%%%%%%%%%%%%%%%%%%%%%%%%%%%%%%%%%%%%%%%%%%%%%%%%%%%%%%%%%%%%%%%%%%
%%%%%%%%%%%%               NEW SUBSECTION                   %%%%%%%%%%%%%%%%%
%%%%%%%%%%%%%%%%%%%%%%%%%%%%%%%%%%%%%%%%%%%%%%%%%%%%%%%%%%%%%%%%%%%%%%%%%%%%%
%%%!!!!!!!!!!!!!!!!!!!!!!!!!!!!!!!!!!!!!!!!!!!!!!!!!!!!!!!!!!!!!!!!!!!!!!!%%%
\subsection{Overview on the 114th Congress session in the U.S. Senate}\label{app:overview}

There were no changes in Senate membership during the 114th Congress session. Hence, speeches were given by 100 Senators. Republican Senator James Risch delivered less than 24 speeches and was excluded from further analysis. The demographics of the remaining 99 Senators can be summarized in the following way: There are 53 Republicans and 44 Democrats, the remaining 2 are independent Senators. The Senate is dominated by men (79 out of 99). Dividing the USA into 5 regions, we obtain the following counts: Midwest (24), West (25), Northeast (22), Southeast (18) and South (10). \citet{Skelley2023} points out that the Congress gets older than ever, which is in line with the generation composition: Silent generation (1928--1945, 22), Boomers (1946--1964, 62) and Gen X (1965--1980, 15). 13 Senators were newcomers to the Senate, 6 of them were even completely new to the Congress and are referred to as ``Freshman''; 30 Senators had already served in more than 10 Congress sessions and are classified as ``Experienced''; the remaining 63 ones are classified as ``Advanced''. 

Religion  has a large set of categories once Christianity is split by different churches. We decided to keep all categories having at least 7 members and combine the rest into one baseline category ``Other'' (18 Senators: Unspecified/Other Protestant, Anglican/Episcopal, Congregationalist, Nondenominational Christian, Buddhist or those who refused to state their religion). The kept categories include: Catholic (25), Presbyterian (14), Baptist (10), Jewish (9), Methodist (9), Lutheran (7), Mormon (7).

\subsection{Important speeches for selected topics}\label{app:speeches}

The most influential speeches are found in a two-step procedure. First, the posterior mean of the document intensities $\theta_{dk}$, i.e.,  $\E \theta_{dk} = \dfrac{\phi_{\theta dk}^\shp}{\phi_{\theta dk}^\rte}$, is evaluated. Given topic $k$, only the subset of documents (of chosen size, e.g., $|\mathcal{B}|$) maximizing these posterior means is considered in the following computations.

In the second step, a quantity inspired by the likelihood-ratio test statistic is evaluated. For topic $k$, we wish to explore which documents are responsible for the ideological discrepancies, that is, to ``test'' $\ideal_{\cdot, k} = \bm 0$. Under the alternative, we have 
$$
\lambda_{dv}^{1} = \sum\limits_{k'=1}^K \theta_{dk'} \beta_{k'v} \exp\left\{\ideal_{a_d k'} \eta_{k'v}\right\},
$$
but under the null,
$$
\lambda_{dv}^{0,k} = \sum\limits_{k'=1; k'\neq k}^K \theta_{dk'} \beta_{k'v} \exp\left\{\ideal_{a_d k'} \eta_{k'v}\right\} + \theta_{dk} \beta_{kv} 
= 
\lambda_{dv}^{1} + \underbrace{\theta_{dk} \beta_{kv} \left(1 - \exp\left\{\ideal_{a_d k} \eta_{kv}\right\}\right)}_{\lambda_{dv}^{\mathsf{dif},k}}.
$$
Then, the log-likelihood ratio test statistic that assesses the contribution of document~$d$ to non-zero ideology for topic~$k$ is approximated by
$$
\chi_{dk} = -2 \sum \limits_{v=1}^V \log \dfrac{p(y_{dv} | \lambda_{dv}^{0,k})}{p(y_{dv}|\lambda_{dv}^{1})} 
= 
2 \sum \limits_{v=1}^V \left[y_{dv} \log \dfrac{\lambda_{dv}^{1}}{\lambda_{dv}^{1} + \lambda_{dv}^{\mathsf{dif},k}} + \lambda_{dv}^{\mathsf{dif},k} \right],
$$
where the $\lambda$ parameters are computed by plugging-in the posterior mean estimates of the individual parameters. 
%Unfortunately, re-estimation under $\ideal_{\cdot, k} = \bm 0$ for every $k$ would lead to superfluous computational costs, hence, we settle with this crude approximation. 
We then order the documents in the batch by $\chi_{dk}$ and consider only those with the highest values.

Below, we include excerpts from some of the speeches identified in this way because they were ranked among the top 20. In addition we also provide some information about the author and their estimated ideological position for that topic.  

\subsubsection*{Topic 4 (Commemoration and Anniversaries)}

Richard Durbin (Democrat, IL, 20150304, $\ideal_{\text{RD},4} \doteq -1.3$):
\textit{
Mr. President, today is the 150\emph{th anniversary} of the second inaugural address of President \emph{Abraham Lincoln}. 
\ldots
It was in the same room on January 20. 2009, that a newly inaugurated President \emph{Barack Obama} signed his first official documents as President of the \emph{United States}. And it was in this room that \emph{Abraham Lincoln} worked long into the night before his second inauguration.
}

Richard Durbin (Democrat, IL, 20150709, $\ideal_{\text{RD},4} \doteq -1.3$):
\textit{
Mr. President, this week the \emph{United States} held a special ceremony to commemorate one of the longest wars in our Nations history the \emph{Vietnam war}. It was a ceremony to honor the \emph{men and women} who served in that long and searing conflict, especially the more than 58.000 young Americans who did not come home from the battle. The Congressional ceremony was held to commemorate what organizers, including the Department of Defense, call the 50\emph{th anniversary} of the \emph{Vietnam war}.
}

Mitch McConnell (Republican, KY, 20151214, $\ideal_{\text{MC},4} \doteq 0.4$):
\textit{
Mr. President, I wish to \emph{pay tribute} today to a distinguished airman and honored Kentuckian who has given over four decades of his life to military service. Maj. Gen. Edward W. Tonin, for 8 years the adjutant general of the Commonwealth of Kentucky, retired from service on General Tonini is a career Air \emph{National Guard} officer and was appointed adjutant general by the former Governor in 2007.
}

Kelly Ayotte (Republican, NH, 20150113, $\ideal_{\text{KA},4} \doteq 0.4$):
\textit{
Mr. President, I wish to recognize the exceptional service and the extraordinary life of Lt. Col. Stephanie Riley of Concord, NH. Born and raised in Henniker, NH. Stephanie graduated from Henniker \emph{High School} in 1984. An excellent student. Stephanie attended St. Pauls advance studies program the summer before her senior year and was the valedictorian of her \emph{high school} class. In 1988, she graduated cum laude from Boston Colleges School of Nursing and in 1989 was commissioned into the U.S. \emph{Air Force}, where she completed a 4month nursing internship at Travis \emph{Air Force} Base in California. Following her internship, she was stationed at the Barksdale \emph{Air Force} Base in Louisiana for the remainder of her 3year tour. 
\ldots 
Showing both her love for the military and her home State, she returned to New Hampshire in 2000 and joined the U.S. \emph{Air Force} Reserves in Westover, MA, and then the NH Air \emph{National Guard} in 2003.
}

\subsubsection*{Topic 9 (Veterans and Health Care)}

John Isakson (Republican, GA, 20151217, $\ideal_{\text{JI},9} \doteq -1.5$):
\textit{
Then we have set goals for next year, a full implementation of the \emph{Veterans Choice} Program and a consolidation of all veterans benefits and VA benefits to see to it that veterans get timely appointments and good-quality services from the physicians in the VA or physicians in their communities.
\ldots
We are going to ensure access to \emph{mental health} so no veteran who finds himself in trouble doesn't have immediate access to counsel. On that point, I commend the \emph{Veterans Administration} for the hotline. The \emph{suicide prevention} hot-line that they established has helped to save lives in this country this year.
}

John Barrasso (Republican, WY, 20160714, $\ideal_{\text{JB},9} \doteq -2.3$):
\textit{
Republicans passed legislation to repeal the Presidents \emph{health care law}. Why? So we can replace it with \emph{health care} reforms that work for the American people. We want to act, and we acted to protect the American people from a \emph{health care law} that has harmed so many people across the country and that so many people feel has absolutely punished them. President \emph{Obama} vetoed the legislation, and Democrats in Congress resisted every attempt to undo years of damage caused by \emph{ObamaCare}.
}

\subsubsection*{Topic 11 (Climate Change)}

Sheldon Whitehouse (Democrat, Boomer generation, RI, 20150721, $\ideal_{\text{SW},11} \doteq -2.3$):
\textit{
This is the 107th time I have come to the floor to urge my colleagues to wake up to the threat of \emph{climate change}. All over the \emph{United States}, State by State by State, we are already seeing the real effects of \emph{carbon pollution}. We see it in our atmosphere, we see it in our oceans, and we see it in our weather, in habitats, and in species. The American people see it. Two thirds of Americans, including half of Republicans, favor government action to reduce \emph{global warming}, and two thirds,  including half of Republicans, would be more likely to vote for a candidate who campaigns on fighting \emph{climate change}.
}

Sheldon Whitehouse (Democrat, Boomer generation, RI, 20160209, $\ideal_{\text{SW},11} \doteq -2.3$):
\textit{
Investigative author Jane Mayer has written an important piece of journalism, her new book, "Dark Money" - about the secret but massive influence buying rightwing billionaires led by the infamous \emph{Koch brothers}.
\ldots
The story of dark money and the story of \emph{climate change} denial are the same story - two sides of the same coin.
\ldots
This machinery of phony science has been wrong over and over. It was wrong about tobacco, wrong about lead paint, wrong about ozone, wrong about mercury, and now it is wrong about \emph{climate change}. They are the same organizations, the same strategies, the same funding sources.
}

Sheldon Whitehouse (Democrat, Boomer generation, RI, 20160712, $\ideal_{\text{SW},11} \doteq -2.3$):
\textit{
The Thomas Jefferson Institute prominently displays a statue of Jefferson on its Web page and claims to be a nonpartisan supporter of "environmental stewardship", but the institute is an outspoken critic of the Presidents \emph{Clean Power Plan} and \emph{renewable sources} of energy and actively sows doubt about climate science.
}

\subsubsection*{Topic 13 (Planned Parenthood and Abortion)}

Dianne Feinstein (Democrat, Jewish, CA, 20150924, $\ideal_{\text{DF},13} \doteq -1.0$):
\textit{
I rise once again to speak against this callous misguided effort to \emph{defund Planned Parenthood}. This is a clear case of politics being put ahead of the country's best interests. This time the majority has tied this effort to the funding of the entire Federal Government, they are willing to \emph{shut down the government} over this issue. That is preposterous. \emph{Planned Parenthood} serves some of the most vulnerable women in our society. It cares for 2.7 million patients in the U.S.- million patients worldwide.
\ldots
Community \emph{health centers} and clinics do great work, but if 2.7 million \emph{Planned Parenthood} patients were suddenly without a doctor, they simply could not handle the sudden influx of new patients.
\ldots
We are standing up for \emph{Planned Parenthood} because we stand up for women. I urge a "no" vote. Lets defeat this bill and move on so we can fund the government and address many other critical issues.
}

Debbie Stabenow (Democrat, Methodist, MI, 20150730, $\ideal_{\text{DS},13} \doteq -1.1$):
\textit{
In my State, 40 percent of the \emph{Planned Parenthood} \emph{health clinics} are located in areas we call medically underserved. There isn't access to other kinds of clinics or \emph{health care}. There may not be a hospital nearby or there may not be many doctors nearby. We are talking about basic \emph{health care}. Unfortunately, we see politics played with women's \emph{preventive health care} and \emph{family planning} over and over again in attacks on \emph{Planned Parenthood}.
\ldots
Fortunately, the vast majority of the American people recognize the value of having \emph{health clinics} like \emph{Planned Parenthood} that are dedicated to serving women's \emph{health care} needs in every community across the country. That is why a poll shows that 64 percent of voters oppose the move by congressional Republicans to \emph{defund Planned Parenthood} and therefore \emph{preventive health care} services such as mammograms, cancer screenings, blood pressure checks, and access to \emph{birth control}.
}

Kelly Ayotte (Republican, Catholic, NH, 20150922, $\ideal_{\text{KA},13} \doteq -0.2$):
\textit{
An issue has come up that is a very important issue. And that is an organization called \emph{Planned Parenthood} and holding \emph{Planned Parenthood} accountable in the wake of deeply disturbing videos that discuss the appalling practice of harvesting the organs and body parts of \emph{unborn babies}.
I was just sick to see the contents of recent videos that have been disclosed that show a callous disregard by officials at \emph{Planned Parenthood} for the dignity of \emph{human life}.
\ldots 
That is why last month I joined a bipartisan majority of Senators in voting to redirect Federal funding from \emph{Planned Parenthood} to community \emph{health centers} that provide \emph{womens health services}.
\ldots
So in good conscience, while I fully support redirecting the money from \emph{Planned Parenthood} to community \emph{health centers} who serve women, I cannot in good conscience participate again in this process, one that would ensure we come closer to the brink of a shutdown, when I have not heard a strategy for success.
}

Charles Grassley (Republican, Baptist, IA, 20150921, $\ideal_{\text{CG},13} \doteq 1.1$):
\textit{
Mr. President, the Popes visit this week to our Nations Capital reminds us all of how very important it is to show compassion and concern for the most innocent and vulnerable among us. \emph{Unborn children} who fall into this category are entitled to the same dignity all human beings share. This is true even when their presence might be uncomfortable or create difficulties, the Pope reminds us. We are now considering moving to a bill known as the Pain-Capable \emph{Unborn Child Protection Act}. This legislation would make no change to our abortion policy in the first 20 weeks of pregnancy. At 20 weeks of fetal age, when the \emph{unborn child} can detect and respond to painful stimuli, the bill would impose some restrictions on elective late-term abortions.
\ldots
I say to them, if you do not support restrictions on abortion after the fifth month of pregnancy, when some babies born prematurely at this stage now are surviving long-term, then what exactly is your limit on abortion? Scientists say the \emph{unborn child} can \emph{feel pain} perhaps even as early as 8 weeks and most certainly by 20 weeks fetal age. The American people overwhelmingly support restrictions on late-term abortions.
\ldots
I urge my colleagues to embrace the sanctity of \emph{human life} and vote to move to this bill so it can at least be considered.
}

Ted Cruz (Republican, Baptist, TX, 20150928, $\ideal_{\text{TC},13} \doteq 1.8$):
\textit{
\emph{Planned Parenthood} is a private organization. It is not even part of the government. But it happens to be politically favored by President Obama and the Democrats. \emph{Planned Parenthood} is also the subject of multiple criminal investigations for being caught on tape apparently carrying out a pattern of ongoing felonies.
\ldots
I am advocating that we fight on things that matter. Don't give \$500 million to \emph{Planned Parenthood}, a corrupt organization that is taking the lives of vast numbers of \emph{unborn children} and selling their body parts.
}

\subsubsection*{Topic 16 (Immigration and Department of Homeland Security)}

Mazie Hirono (Democrat, HI (West), Advance experience, 20150127, $\ideal_{\text{MH},16} \doteq -1.4$):
\textit{
I rise today on the important issue of \emph{funding the Department of Homeland Security} and to urge my colleagues to come together and pass a clean appropriations bill with regard to this agency. The Department of \emph{Homeland Security}, or DHS, is charged with \emph{border security} and immigration enforcement.
\ldots
Now is the time when we should be working together on commonsense and \emph{comprehensive immigration reform} that the vast majority of Americans support. \emph{Comprehensive immigration reform} is supported by 70 percent of the American people. In the past Congress, nearly 70 percent of the Senate supported our bipartisan immigration bill. Our bipartisan bill was a compromise. It strengthened \emph{border security}, modernized our system, addressed visa backlogs, and allowed millions of undocumented people to step out of the shadows, get in line, and work toward becoming American citizens.
\ldots
Recklessly shutting down the Department of \emph{Homeland Security} will not fix our \emph{broken immigration system}.
}

Edward Markey (Democrat, MA (Northeast), Experienced, 20150225, $\ideal_{\text{EM},16} \doteq -0.8$):
\textit{
Mr. President, everyone agrees that our \emph{immigration system is broken}. The \emph{immigration system} we have now hurts our economy, and it hurts our \emph{national security}. The Senate passed a bipartisan immigration bill, the House of Representatives chose not to act. Again, the Senate passed a \emph{comprehensive immigration} bill. That is why I supported the Executive action by President Obama to address our immediate immigration crisis. We cannot wait for the House of Representatives Republicans to act, and that is because immigration is one of our country's greatest strengths. Immigrants are a vital part of the fabric of Massachusetts and of our country. They start businesses, they create jobs, and they contribute to our communities. The Presidents Executive order recognizes the value of immigrants to our country. 
\ldots
What unites us in Massachusetts and all across America is the unshakable belief that no matter where you come from, no matter what your circumstances, you can achieve the American dream. The \emph{immigration system} we have now doesn't reflect those values. Unfortunately, instead of working to fix the problems with our \emph{immigration system}, the majority of the Senate has been manufacturing a government shutdown of the Department of \emph{Homeland Security}.
\ldots
We should not have any question raised about the Department of \emph{Homeland Security} being on the job protecting our citizens and providing the security our country needs. That is where we are right now, and the Republicans are holding up the funding of this vital agency under the misguided notion that they are going to be able to write the entire \emph{comprehensive immigration} bill inside a Department of \emph{Homeland Security} budget. It is not going to happen. Everyone in this country knows it is not going to happen. The Republicans are playing a dangerous game with the security of our country. I ask all who make the decisions in the Republican Party to please tell their most radical Members that the Department of \emph{Homeland Security} must be funded. It must be funded this week. We must not only pay those who work for us, but we should thank them every day for the security they provide to our country.
}

Jefferson Sessions (Republican, AL (Southeast), Advanced experience, 20150130, $\ideal_{\text{JS},16} \doteq 1.6$):
\textit{
Colleagues, this is so serious that the \emph{Immigration and Customs Enforcement} officials, their association filed a lawsuit, and they challenged the actions of their supervisors telling them not to enforce plain \emph{immigration law}.
\ldots
I will be a little bit courteous at this point and just quote some of the statements from all separate Democratic Senators in the last few months when asked about this Executive amnesty by the President. A lot of Senators have never been asked. They are probably thankful they weren't asked. This is what one Senator said: ... but the President shouldn't make such a significant policy change on his own. Another Democratic Senator: ... but executive orders aren't the way to do it. Another Senator: I disagree with the Presidents decision to use executive action to make changes to our \emph{immigration system}. Another Democratic Senator: I'm disappointed the President decided to use executive action at this time on this issue, as it could poison any hope of compromise or bipartisanship in the new Senate before it has even started. Its Congress job to pass legislation and deal with issues of this magnitude. Absolutely correct. It is Congress's duty to do this. What about another Democratic Senator: I worry that his taking unilateral action could in fact inflame public opinion, change the subject from immigration to the President. I also have constitutional concerns about where prosecutorial discretion ends and unconstitutional authority begins. A wise quote. I think. Another Senator: I have concerns about executive action... This is a job for Congress. and its time for the House to act. Another Democratic Senator: ... the best way to get a comprehensive solution is to take this through the legislative process. So I would say, colleagues, why would any Senator, Democrat or Republican, when the very integrity of the constitutional powers given to Congress are eroded in a dramatic way by the President of the \emph{United States} not want to assert congressional authority?
\ldots
People are coming from abroad. They want to come to America. We have always had the most generous \emph{immigration system} in the world, and we believe in immigration. But they should come lawfully and the Congress should help create a system that supports a lawful entry into America. The council that represents the \emph{Customs and Immigration Service} Officers just January 22nd of this year issued a strong statement.
\ldots
They said: The dedicated immigration service officers and adjudicators at USCIS are in desperate need of help. The \emph{Presidents executive amnesty} order for 5 million \emph{illegal immigrants} places the mission of USCIS in grave peril.
\ldots
The Democrats are saying: we are not even going to go to this bill that would fund \emph{Homeland Security}. And if we don't go to it, then \emph{Homeland Security} is not funded. Are they going to block a bill that would fund \emph{Homeland Security}? Senator McConnell is saying you can have your relevant amendment. If you don't like the language the House put in that says the money can only go to fund lawful activities.
}

Charles Grassley (Republican, IA (Midwest), Experienced, 20151019, $\ideal_{\text{CG},16} \doteq 1.6$):
\textit{
We will move to take up the Stop Sanctuary Policies and Protect Americans Acta bill that should put an end to sanctuary jurisdictions, give law enforcement important tools they need to detain criminals, and increase penalties for dangerous and repeat offenders of our \emph{immigration laws}.
\ldots
America saw these policies play out in July when Kate Steinle was innocently killed while walking along a San Francisco pier with her father. The murderer, who was illegally in the country and actually deported five times prior to that day, was released into the community by a sanctuary jurisdiction that did not honor the detainer issued by \emph{Immigration and Customs Enforcement}.
\ldots
We heard from Mrs. Susan Oliver. She is the widow of Deputy Danny Oliver. This is the family. He was a police officer in Sacramento. CA. Danny was killed while on duty by an \emph{illegal immigrant} who was previously arrested on two separate occasions for drug-related charges and twice deported.
\ldots
Moreover, in June of this year, the administration rolled out a new program that reduces the enforcement priorities and announced it would not seek the custody of many criminals who are in the country illegally. This is called the Priority Enforcement Program, PEP for short. That program actually gives sanctuary jurisdictions permission to continue ignoring \emph{Immigration and Customs Enforcement} detainers. PEP even discourages compliant jurisdictions from further cooperation with \emph{Immigration and Customs Enforcement} because it now only issues detainers for individuals who are already convicted of certain crimes deemed priorities by the Department of \emph{Homeland Security}.
}

\subsubsection*{Topic 24 (Export, Import, and Business)}

Jon Tester (Democrat, MT, 20151007, $\ideal_{\text{JT},24} \doteq -1.6$):
\textit{
And the \emph{Land and Water Conservation Fund} is a big reason for that. Montana's outdoor economy brings in nearly \$6 billion a year. Let me say that again. The outdoor economy, supported by the \emph{Land and Water Conservation Fund}, brings in nearly \$6 billion a year.
\ldots
While Montanans certainly benefit from the fund, there are \emph{Land and Water Conservation Fund} projects in nearly every county of the \emph{United States}. Yes, this fund is responsible for protecting prime hunting and fishing, but it is also responsible for building trails and improving parks, playgrounds, and ball fields in every State in the country.
}

Heidi Heitkamp (Democrat, ND, 20151105, $\ideal_{\text{HH},24} \doteq -0.7$):
\textit{
I wish to take a few moments today to talk about \emph{small business}, to talk about the people who have been dramatically affected by the closure of the \emph{ExIm Bank} and why it is so important that we understand, appreciate, and not have a long-term history that does not move the \emph{ExIm Bank} forward.
\ldots
The \emph{ExportImport Bank} has been closed for over 3 months, preventing needed support for \emph{small business} across the country.
\ldots
We believe we are ready, willing, and excited about the opportunity of once again opening the doors of the \emph{ExportImport Bank} and welcoming American business in and saying once again. "America is open for business" to the rest of the world.
}

Debbie Stabenow (Democrat, MI, 20150519, $\ideal_{\text{DS},24} \doteq 1.2$):
\textit{
What we believe on behalf of our constituents, the people we represent, are the most important things that we want to make sure are covered: enforcement, strong labor and environmental standards, and the No. 1 trade distorting policy in the world today, which is \emph{currency manipulation}. We want to be able to say, if you are going to get this special ability to take away our right to change something, then we expect certain things. We expect that we are going to be negotiating from a position of strength so that we are racing up in the world economy, bringing other countries up in terms of wages, what is happening in terms of protecting our environment, protecting our \emph{intellectual property} rights, stopping other countries from cheating on currency or other trade violations.
\ldots
We are saying that if we are going to let you go into a negotiation and come out with a \emph{trade agreement} of 40 percent of the global economy in Asia and where we are seeing the bulk of the \emph{currency manipulation}, then we believe there ought to be an enforceable standard, that we ought to have an expectation of a \emph{currency manipulation} provision that would be enforceable at least as a negotiating objective.
\ldots
makes me very concerned about what is really going on in the \emph{TransPacific Partnership}. I urge my colleagues to join me and Senator Portman in passing this very reasonable amendment to make \emph{currency manipulation} a priority in our negotiations.
}

Orrin Hatch (Republican, UT, 20150521, $\ideal_{\text{OH},24} \doteq 1.7$):
\textit{
I want to take a few minutes today to talk once again about Congress's role in advancing our Nations trade policies and specifically on the increasingly important issues of digital trade and \emph{intellectual property} rights. Lets keep in mind that the last time Congress passed TPA (\emph{Trade Promotion Authority}) was in 2002. We live in a very different world than we did 13 years ago. Technology is vastly different.
\ldots
I would urge any of my colleagues who oppose this bill to explain how they plan to give \emph{American workers} and businesses in the digital economy an opportunity to thrive in an increasingly competitive marketplace - global marketplace.
\ldots
Our bill also calls for the elimination of measures that require U.S. companies to locate their \emph{intellectual property} abroad as a market access or investment condition.
\ldots
I urge each of my colleagues to work with me to help move this bill forward so we can negotiate strong \emph{trade agreements} that serve today's economy as well as set the stage for Americas next generation of entrepreneurs and innovators.
}

\subsection{Word clouds}\label{app:wordclouds}

\begin{figure}
    \centering
    \includegraphics[width=0.97\textwidth]{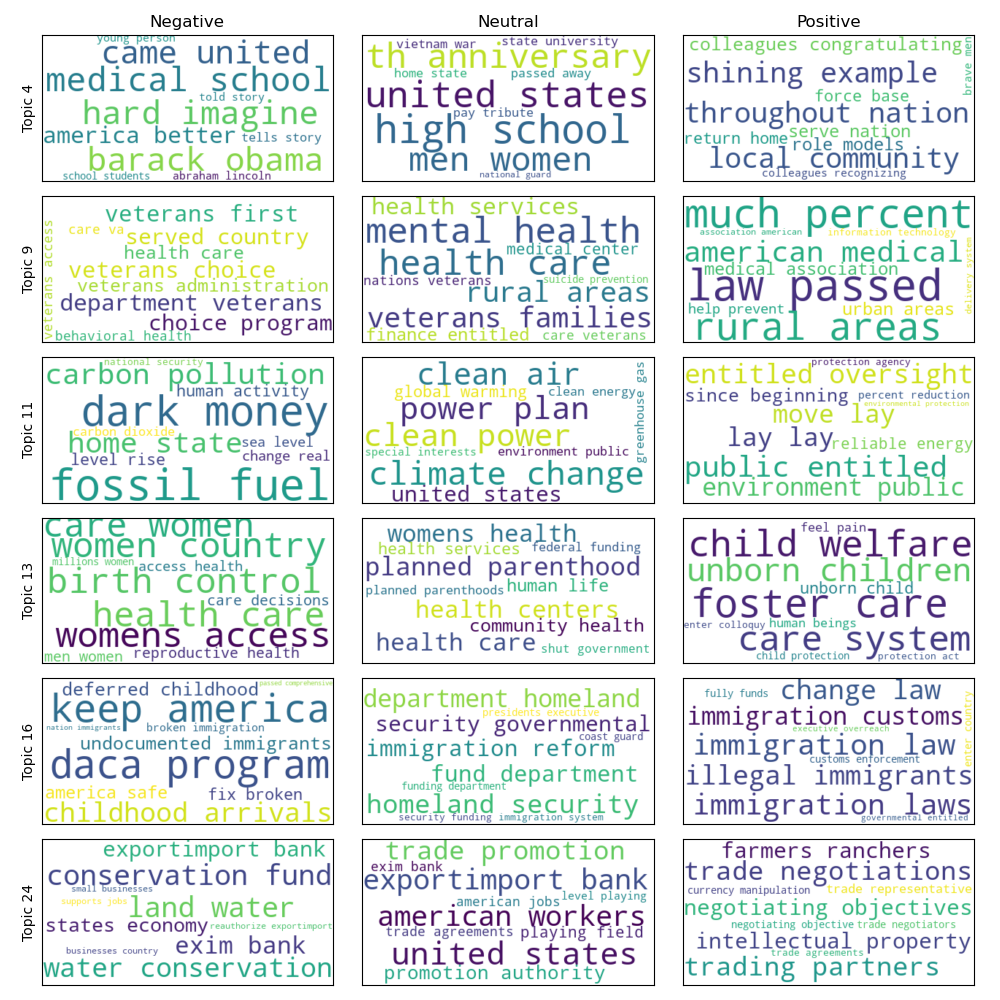}
    \caption{Word clouds for selected topics. The neutral column is based on the variational mean of $\beta_{kv}$. Positive and negative columns correspond to the $\eta_{kv}$ contribution indicating the terms that discriminate between the ideological positions conditional on  $\E \log\beta_{kv} > -1$.}
    \label{fig:TBIPhier_all_wordclouds_justeta_logscale_selected_topics}
\end{figure}

\begin{figure}
    \centering
    \includegraphics[width=0.97\textwidth]{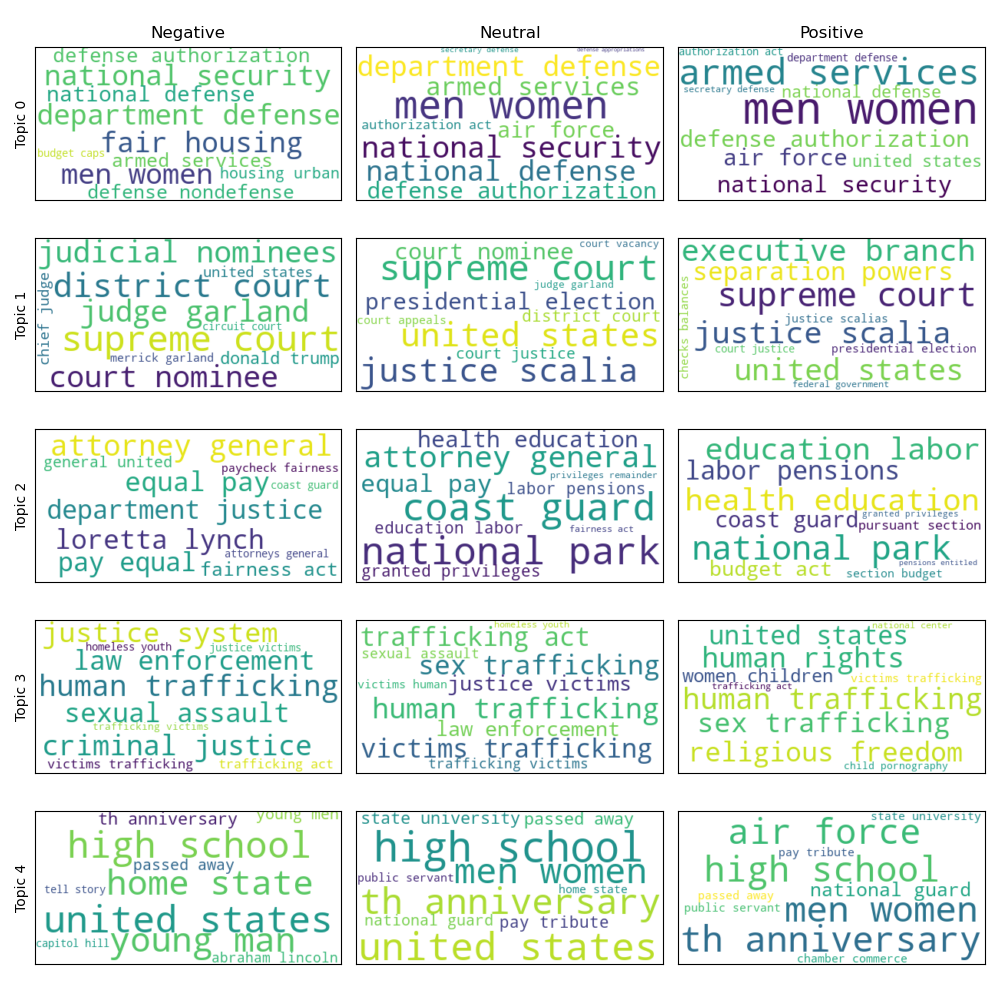}
    \caption{Word clouds for topics 0--4. Positive and negative columns correspond to the variational mean of $\bm \beta$ multiplied by the variational mean of the ideological term for an ideological position equal to 1 or $-1$, respectively.}
    \label{fig:TBIPhier_all_wordclouds_logscale_topics_0_4}
\end{figure}

\begin{figure}
    \centering
    \includegraphics[width=0.97\textwidth]{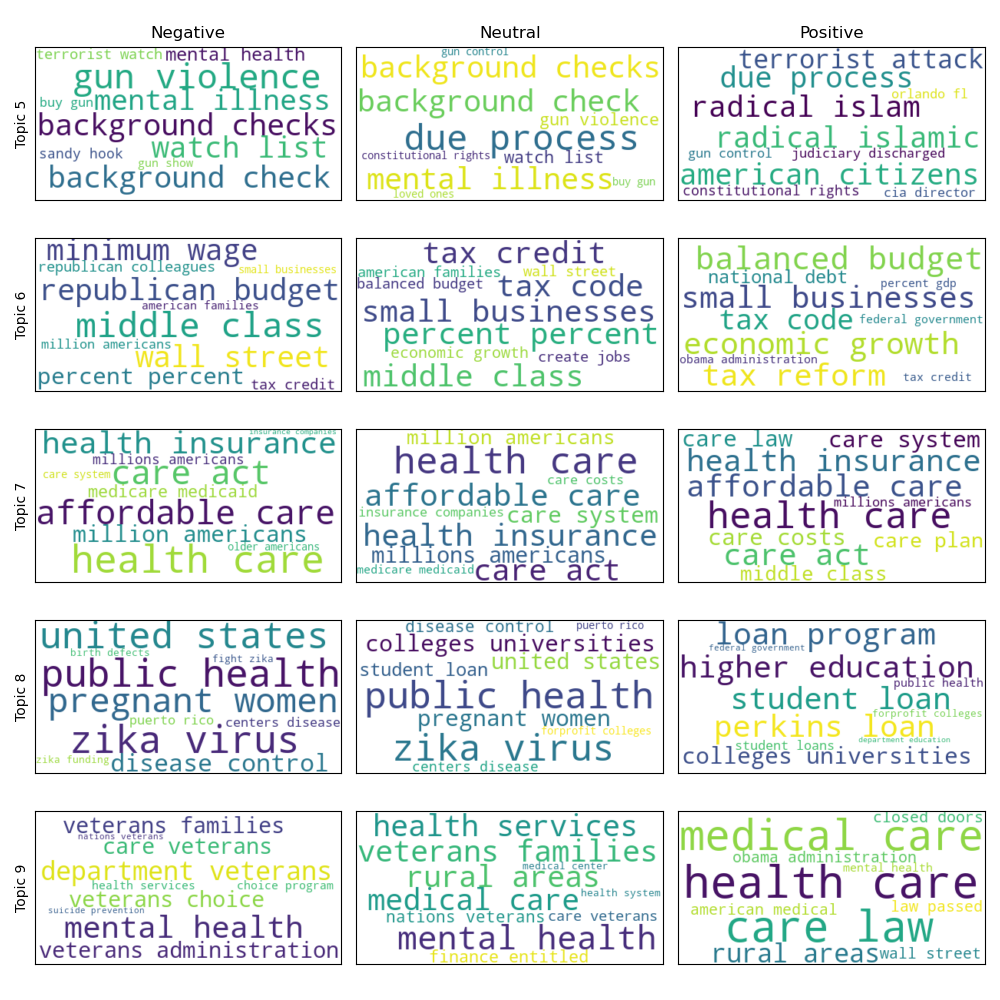}
    \caption{Word clouds for topics 5--9. Positive and negative columns correspond to the variational mean of $\bm \beta$ multiplied by the variational mean of the ideological term for an ideological position equal to 1 or $-1$, respectively.}
    \label{fig:TBIPhier_all_wordclouds_logscale_topics_5_9}
\end{figure}

\begin{figure}
    \centering
    \includegraphics[width=0.97\textwidth]{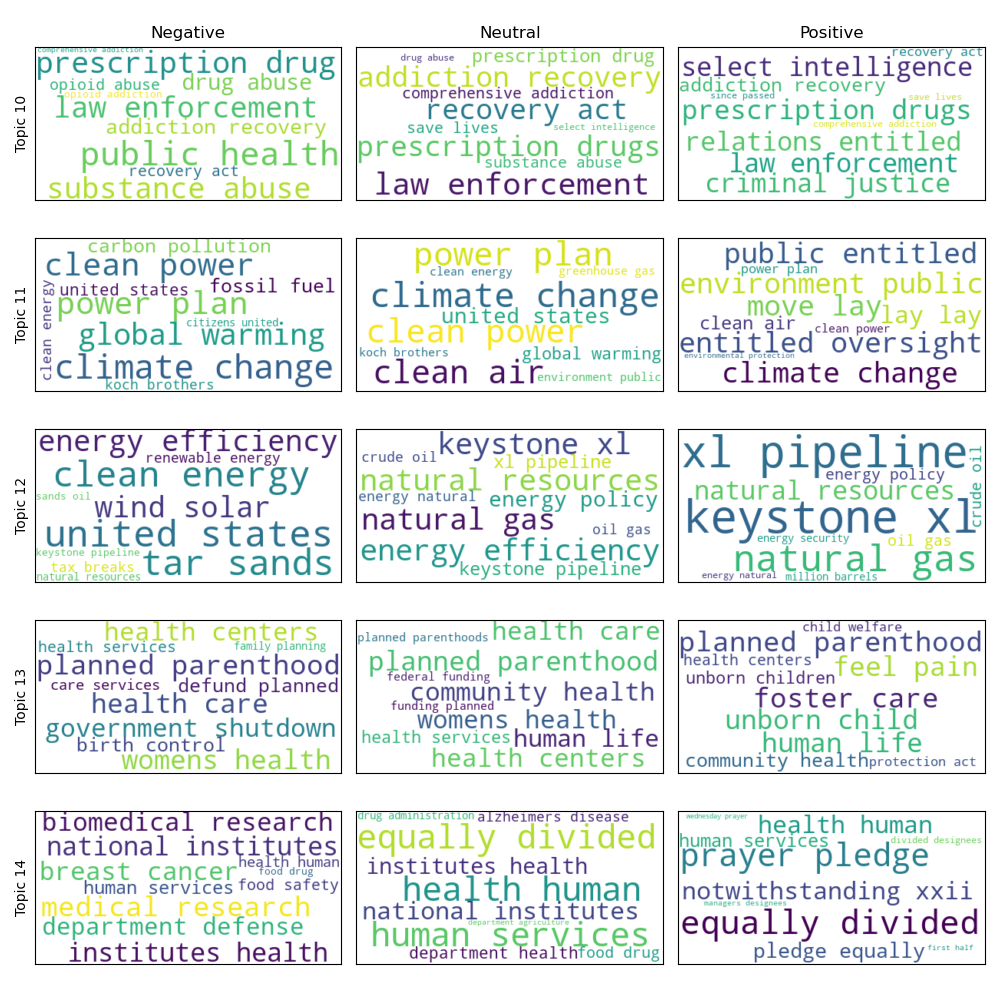}
    \caption{Word clouds for topics 10--14. Positive and negative columns correspond to the variational mean of $\bm \beta$ multiplied by the variational mean of the ideological term for an ideological position equal to 1 or $-1$, respectively.}
    \label{fig:TBIPhier_all_wordclouds_logscale_topics_10_14}
\end{figure}

\begin{figure}
    \centering
    \includegraphics[width=0.97\textwidth]{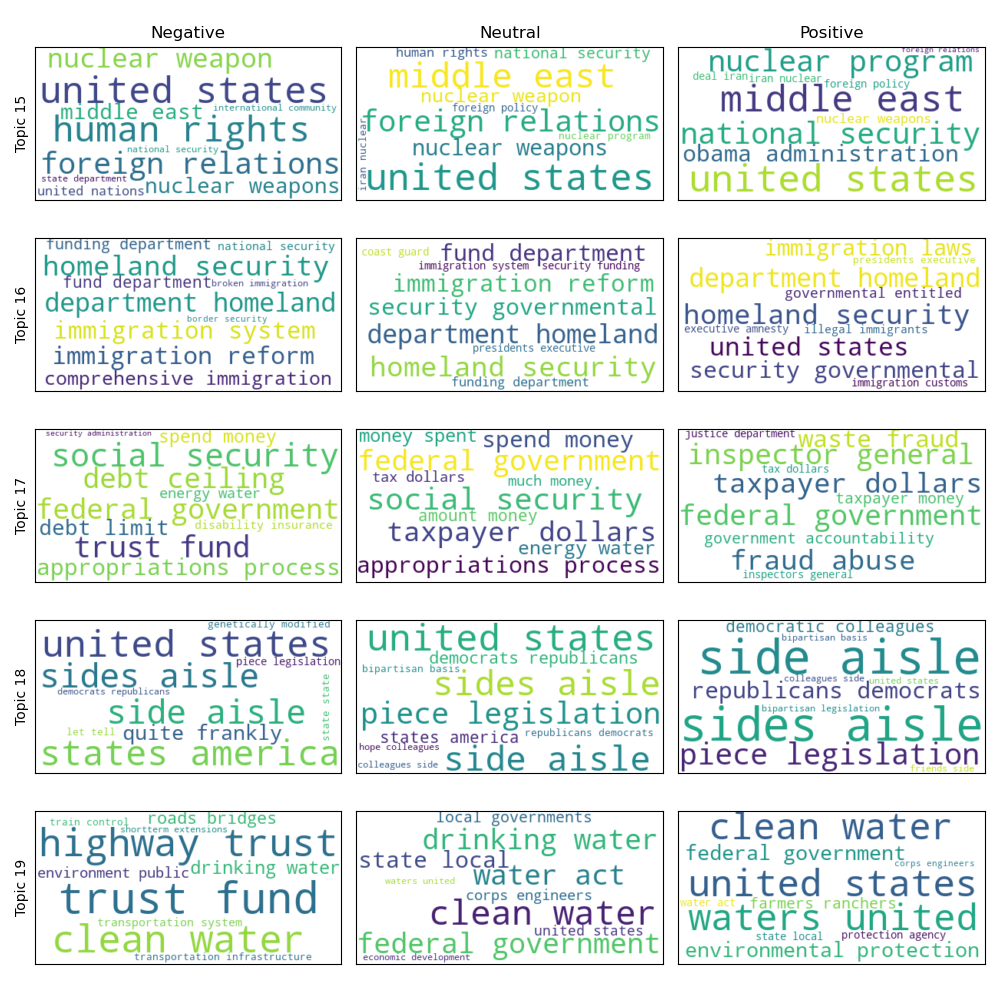}
    \caption{Word clouds for topics 15--19. Positive and negative columns correspond to the variational mean of $\bm \beta$ multiplied by the variational mean of the ideological term for an ideological position equal to 1 or $-1$, respectively.}
    \label{fig:TBIPhier_all_wordclouds_logscale_topics_15_19}
\end{figure}

\begin{figure}
    \centering
    \includegraphics[width=0.97\textwidth]{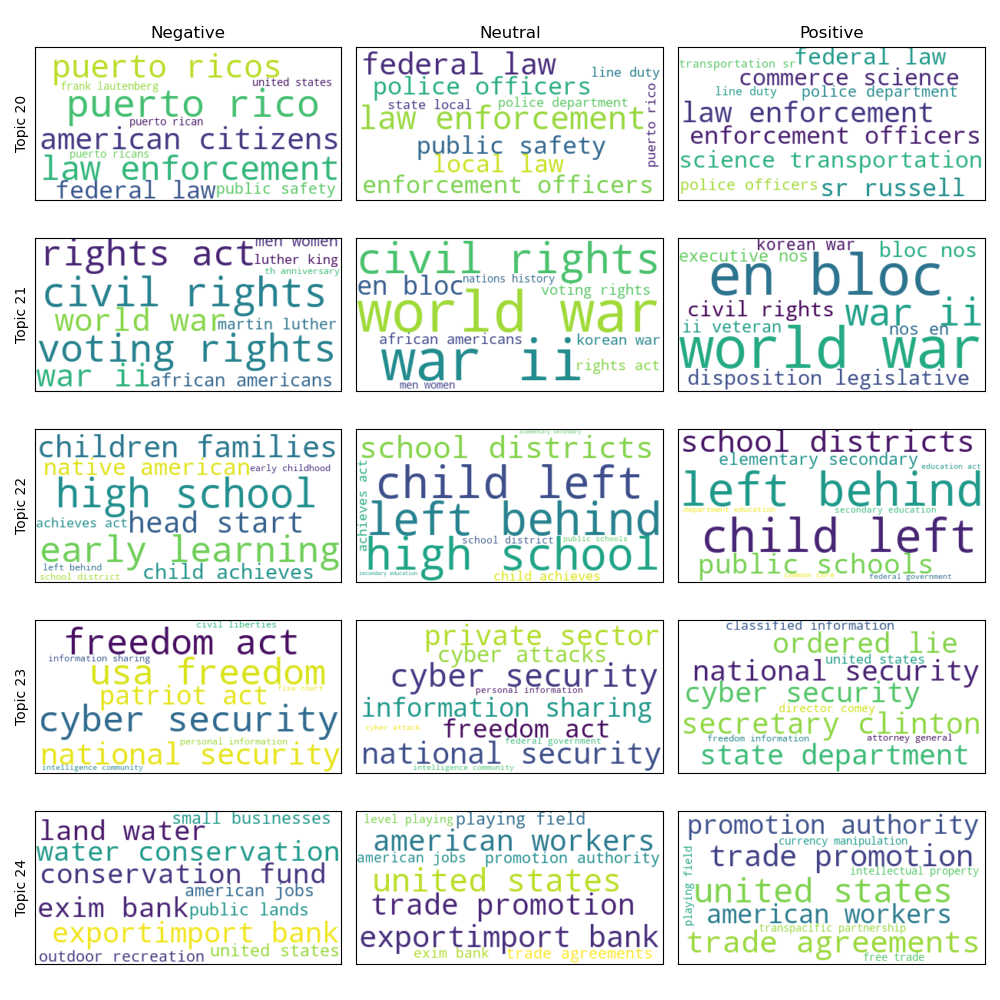}
    \caption{Word clouds for topics 20--24. Positive and negative columns correspond to the variational mean of $\bm \beta$ multiplied by the variational mean of the ideological term for an ideological position equal to 1 or $-1$, respectively.}
    \label{fig:TBIPhier_all_wordclouds_logscale_topics_20_24}
\end{figure}

\clearpage
%%%!!!!!!!!!!!!!!!!!!!!!!!!!!!!!!!!!!!!!!!!!!!!!!!!!!!!!!!!!!!!!!!!!!!!!!!%%%
%%%%%%%%%%%%%%%%%%%%%%%%%%%%%%%%%%%%%%%%%%%%%%%%%%%%%%%%%%%%%%%%%%%%%%%%%%%%%
%%%%%%%%%%%%               NEW SUBSECTION                   %%%%%%%%%%%%%%%%%
%%%%%%%%%%%%%%%%%%%%%%%%%%%%%%%%%%%%%%%%%%%%%%%%%%%%%%%%%%%%%%%%%%%%%%%%%%%%%
%%%!!!!!!!!!!!!!!!!!!!!!!!!!!!!!!!!!!!!!!!!!!!!!!!!!!!!!!!!!!!!!!!!!!!!!!!%%%
\subsection{Regression coefficients}\label{app:coefs}

\renewcommand{\arraystretch}{0.75}

\begin{table}[h!]
    \centering
    \resizebox{0.5\textwidth}{!}{
        \input{STBS_a_regression_coefs.tex}
    }
    \caption{Posterior estimates of the regression coefficients for the model with fixed ideological positions across topics and including no interactions.}
    \label{tab:regression_coefs}
\end{table}

\begin{comment}
\begin{table}
    \centering
    \resizebox{0.5\textwidth}{!}{
        \input{JASA_paper/tab/TBIPhier_all/regression_coefs_k_13.tex}
    }
    \caption{Posterior means of the regression coefficients estimated for the model including interactions.}
    \label{tab:regression_coefs_k_13}
\end{table}

\begin{table}
    \centering
    \resizebox{0.8\textwidth}{!}{
        \input{JASA_paper/tab/TBIPhier_all/regression_coefs_k_5_9.tex}
    }
    \caption{Posterior means of the regression coefficients estimated for the model including interactions, 2 topics next to each other.}
    \label{tab:regression_coefs_k_5_9}
\end{table}
\end{comment}

\begin{sidewaystable}
    \centering
    \resizebox{\textwidth}{!}{
        \input{STBS_ak_regression_coefs_k_4_9_11_13_16_24.tex} 
    }
    \caption{Posterior estimates of the regression coefficients for the model with topic-specific ideological positions for selected topics and including interactions.}
    \label{tab:regression_coefs_k_4_9_11_13_16_24}
\end{sidewaystable}

%% file: STBS_a_regression_coefs.tex
\begin{tabular}{ll|rrrr}
\noalign{\smallskip}
\toprule
\multirow{2}{*}{Coefficient} & \multirow{2}{*}{Category}  & \multicolumn{4}{c}{}\\
 &  & Estimate & SE & CCP & CCP (all)\\
\midrule
\noalign{\smallskip}
\texttt{intercept} &  & $-0.535$ & $0.097$ & $<0.001$ & \\
\noalign{\smallskip}
\midrule
\noalign{\smallskip}
\multirow{2}{*}{\texttt{party}} & \texttt{Republican} & $0.734$ & $0.067$ & $<0.001$ & \multirow{2}{*}{$<0.001$}\\
 & \texttt{Independent} & $-0.064$ & $0.171$ & $0.709$ & \\
\noalign{\smallskip}
\midrule
\noalign{\smallskip}
\multirow{1}{*}{\texttt{gender}} & \texttt{Female} & $-0.042$ & $0.069$ & $0.546$ & \\
\noalign{\smallskip}
\midrule
\noalign{\smallskip}
\multirow{4}{*}{\texttt{region}} & \texttt{Midwest} & $0.158$ & $0.082$ & $0.053$ & \multirow{4}{*}{$0.176$}\\
 & \texttt{Southeast} & $0.094$ & $0.089$ & $0.290$ & \\
 & \texttt{South} & $0.115$ & $0.110$ & $0.294$ & \\
 & \texttt{West} & $0.196$ & $0.083$ & $0.019$ & \\
\noalign{\smallskip}
\midrule
\noalign{\smallskip}
\multirow{2}{*}{\texttt{generation}} & \texttt{Boomer} & $-0.098$ & $0.077$ & $0.203$ & \multirow{2}{*}{$0.443$}\\
 & \texttt{Gen X} & $-0.095$ & $0.108$ & $0.376$ & \\
\noalign{\smallskip}
\midrule
\noalign{\smallskip}
\multirow{2}{*}{\texttt{experience}} & \texttt{(1,10]} & $0.159$ & $0.073$ & $0.030$ & \multirow{2}{*}{$0.015$}\\
 & \texttt{(0,1]} & $0.355$ & $0.132$ & $0.007$ & \\
\noalign{\smallskip}
\midrule
\noalign{\smallskip}
\multirow{7}{*}{\texttt{religion}} & \texttt{Catholic} & $-0.030$ & $0.077$ & $0.697$ & \multirow{7}{*}{$0.300$}\\
 & \texttt{Presbyterian} & $0.113$ & $0.090$ & $0.211$ & \\
 & \texttt{Baptist} & $0.160$ & $0.106$ & $0.129$ & \\
 & \texttt{Jewish} & $-0.044$ & $0.105$ & $0.672$ & \\
 & \texttt{Methodist} & $-0.060$ & $0.101$ & $0.554$ & \\
 & \texttt{Lutheran} & $-0.106$ & $0.117$ & $0.364$ & \\
 & \texttt{Mormon} & $-0.123$ & $0.119$ & $0.301$ & \\
\noalign{\smallskip}
\bottomrule
\noalign{\medskip}
\end{tabular}

%% file: STBS_ak_regression_coefs_k_4_9_11_13_16_24.tex
\begin{tabular}{ll|rrrr|rrrr|rrrr|rrrr|rrrr|rrrr}
\noalign{\smallskip}
\toprule
\multirow{2}{*}{Coefficient} & \multirow{2}{*}{Category}  & \multicolumn{4}{c}{Topic 4} & \multicolumn{4}{c}{Topic 9} & \multicolumn{4}{c}{Topic 11} & \multicolumn{4}{c}{Topic 13} & \multicolumn{4}{c}{Topic 16} & \multicolumn{4}{c}{Topic 24}\\
 &  & Estimate & SE & CCP & CCP (all) & Estimate & SE & CCP & CCP (all) & Estimate & SE & CCP & CCP (all) & Estimate & SE & CCP & CCP (all) & Estimate & SE & CCP & CCP (all) & Estimate & SE & CCP & CCP (all)\\
\midrule
\noalign{\smallskip}
\texttt{intercept} &  & $-0.425$ & $0.049$ & $<0.001$ &  & $-0.429$ & $0.049$ & $<0.001$ &  & $-0.447$ & $0.049$ & $<0.001$ &  & $-0.451$ & $0.049$ & $<0.001$ &  & $-0.455$ & $0.049$ & $<0.001$ &  & $-0.451$ & $0.049$ & $<0.001$ & \\
\noalign{\smallskip}
\midrule
\noalign{\smallskip}
\multirow{2}{*}{\texttt{party}} & \texttt{Republican} & $0.283$ & $0.112$ & $0.012$ & \multirow{2}{*}{$0.023$} & $0.015$ & $0.112$ & $0.893$ & \multirow{2}{*}{$0.811$} & $0.720$ & $0.112$ & $<0.001$ & \multirow{2}{*}{$<0.001$} & $0.574$ & $0.112$ & $<0.001$ & \multirow{2}{*}{$<0.001$} & $0.737$ & $0.112$ & $<0.001$ & \multirow{2}{*}{$<0.001$} & $0.212$ & $0.112$ & $0.059$ & \multirow{2}{*}{$0.116$}\\
 & \texttt{Independent} & $0.190$ & $0.168$ & $0.259$ &  & $0.107$ & $0.168$ & $0.526$ &  & $0.164$ & $0.168$ & $0.328$ &  & $0.142$ & $0.168$ & $0.399$ &  & $0.308$ & $0.168$ & $0.067$ &  & $0.151$ & $0.168$ & $0.370$ & \\
\noalign{\smallskip}
\midrule
\noalign{\smallskip}
\multirow{1}{*}{\texttt{gender}} & \texttt{Female} & $0.156$ & $0.093$ & $0.092$ &  & $-0.044$ & $0.093$ & $0.633$ &  & $0.139$ & $0.093$ & $0.134$ &  & $-0.255$ & $0.093$ & $0.006$ &  & $-0.039$ & $0.093$ & $0.674$ &  & $0.073$ & $0.093$ & $0.428$ & \\
\noalign{\smallskip}
\midrule
\noalign{\smallskip}
\multirow{4}{*}{\texttt{region}} & \texttt{Midwest} & $0.020$ & $0.065$ & $0.759$ & \multirow{4}{*}{$0.075$} & $-0.049$ & $0.065$ & $0.456$ & \multirow{4}{*}{$0.034$} & $0.082$ & $0.065$ & $0.207$ & \multirow{4}{*}{$0.141$} & $0.033$ & $0.065$ & $0.618$ & \multirow{4}{*}{$0.093$} & $0.046$ & $0.065$ & $0.479$ & \multirow{4}{*}{$0.258$} & $0.061$ & $0.065$ & $0.350$ & \multirow{4}{*}{$0.346$}\\
 & \texttt{Southeast} & $0.104$ & $0.070$ & $0.137$ &  & $0.116$ & $0.070$ & $0.095$ &  & $0.121$ & $0.070$ & $0.083$ &  & $0.164$ & $0.070$ & $0.019$ &  & $0.086$ & $0.070$ & $0.215$ &  & $0.096$ & $0.070$ & $0.169$ & \\
 & \texttt{South} & $0.109$ & $0.074$ & $0.141$ &  & $0.114$ & $0.074$ & $0.124$ &  & $0.104$ & $0.074$ & $0.160$ &  & $0.073$ & $0.074$ & $0.327$ &  & $0.121$ & $0.074$ & $0.103$ &  & $0.077$ & $0.074$ & $0.296$ & \\
 & \texttt{West} & $0.135$ & $0.063$ & $0.032$ &  & $0.137$ & $0.063$ & $0.029$ &  & $0.068$ & $0.063$ & $0.274$ &  & $0.085$ & $0.063$ & $0.173$ &  & $0.067$ & $0.063$ & $0.284$ &  & $0.068$ & $0.063$ & $0.279$ & \\
\noalign{\smallskip}
\midrule
\noalign{\smallskip}
\multirow{2}{*}{\texttt{generation}} & \texttt{Boomer} & $0.160$ & $0.036$ & $<0.001$ & \multirow{2}{*}{$<0.001$} & $0.148$ & $0.036$ & $<0.001$ & \multirow{2}{*}{$<0.001$} & $0.147$ & $0.036$ & $<0.001$ & \multirow{2}{*}{$<0.001$} & $0.155$ & $0.036$ & $<0.001$ & \multirow{2}{*}{$<0.001$} & $0.146$ & $0.036$ & $<0.001$ & \multirow{2}{*}{$<0.001$} & $0.148$ & $0.036$ & $<0.001$ & \multirow{2}{*}{$<0.001$}\\
 & \texttt{Gen X} & $0.143$ & $0.079$ & $0.069$ &  & $0.255$ & $0.079$ & $0.001$ &  & $0.222$ & $0.079$ & $0.005$ &  & $0.239$ & $0.079$ & $0.002$ &  & $0.214$ & $0.079$ & $0.007$ &  & $0.139$ & $0.079$ & $0.078$ & \\
\noalign{\smallskip}
\midrule
\noalign{\smallskip}
\multirow{2}{*}{\texttt{experience}} & \texttt{(1,10]} & $0.047$ & $0.059$ & $0.424$ & \multirow{2}{*}{$0.680$} & $0.025$ & $0.059$ & $0.668$ & \multirow{2}{*}{$0.808$} & $-0.012$ & $0.059$ & $0.843$ & \multirow{2}{*}{$0.589$} & $0.076$ & $0.059$ & $0.196$ & \multirow{2}{*}{$0.386$} & $0.039$ & $0.059$ & $0.509$ & \multirow{2}{*}{$0.721$} & $-0.012$ & $0.059$ & $0.843$ & \multirow{2}{*}{$0.676$}\\
 & \texttt{(0,1]} & $0.057$ & $0.137$ & $0.677$ &  & $-0.063$ & $0.137$ & $0.644$ &  & $0.136$ & $0.137$ & $0.320$ &  & $0.077$ & $0.137$ & $0.573$ &  & $0.070$ & $0.137$ & $0.609$ &  & $0.116$ & $0.137$ & $0.397$ & \\
\noalign{\smallskip}
\midrule
\noalign{\smallskip}
\multirow{7}{*}{\texttt{religion}} & \texttt{Catholic} & $0.023$ & $0.061$ & $0.708$ & \multirow{7}{*}{$0.040$} & $0.024$ & $0.061$ & $0.700$ & \multirow{7}{*}{$0.012$} & $0.081$ & $0.061$ & $0.187$ & \multirow{7}{*}{$0.069$} & $0.074$ & $0.061$ & $0.224$ & \multirow{7}{*}{$0.006$} & $0.016$ & $0.061$ & $0.794$ & \multirow{7}{*}{$0.040$} & $0.025$ & $0.061$ & $0.682$ & \multirow{7}{*}{$0.065$}\\
 & \texttt{Presbyterian} & $0.207$ & $0.083$ & $0.013$ &  & $0.291$ & $0.083$ & $<0.001$ &  & $0.161$ & $0.083$ & $0.053$ &  & $0.220$ & $0.083$ & $0.008$ &  & $0.205$ & $0.083$ & $0.013$ &  & $0.136$ & $0.083$ & $0.102$ & \\
 & \texttt{Baptist} & $0.030$ & $0.095$ & $0.751$ &  & $-0.012$ & $0.095$ & $0.903$ &  & $-0.007$ & $0.095$ & $0.940$ &  & $0.068$ & $0.095$ & $0.475$ &  & $0.057$ & $0.095$ & $0.551$ &  & $-0.019$ & $0.095$ & $0.843$ & \\
 & \texttt{Jewish} & $0.076$ & $0.103$ & $0.458$ &  & $-0.002$ & $0.103$ & $0.983$ &  & $-0.140$ & $0.103$ & $0.173$ &  & $-0.164$ & $0.103$ & $0.112$ &  & $-0.057$ & $0.103$ & $0.583$ &  & $0.188$ & $0.103$ & $0.068$ & \\
 & \texttt{Methodist} & $-0.155$ & $0.091$ & $0.090$ &  & $-0.087$ & $0.091$ & $0.337$ &  & $-0.078$ & $0.091$ & $0.389$ &  & $-0.139$ & $0.091$ & $0.126$ &  & $-0.094$ & $0.091$ & $0.302$ &  & $-0.008$ & $0.091$ & $0.933$ & \\
 & \texttt{Lutheran} & $-0.174$ & $0.084$ & $0.039$ &  & $-0.173$ & $0.084$ & $0.040$ &  & $-0.160$ & $0.084$ & $0.057$ &  & $-0.166$ & $0.084$ & $0.048$ &  & $-0.176$ & $0.084$ & $0.037$ &  & $-0.157$ & $0.084$ & $0.063$ & \\
 & \texttt{Mormon} & $0.072$ & $0.140$ & $0.605$ &  & $-0.011$ & $0.140$ & $0.937$ &  & $-0.145$ & $0.140$ & $0.302$ &  & $-0.171$ & $0.140$ & $0.223$ &  & $-0.203$ & $0.140$ & $0.148$ &  & $-0.254$ & $0.140$ & $0.070$ & \\
\noalign{\smallskip}
\midrule
\noalign{\smallskip}
\multirow{1}{*}{\texttt{party\_Republican:gender}} & \texttt{Female} & $-0.043$ & $0.092$ & $0.638$ &  & $-0.141$ & $0.092$ & $0.126$ &  & $-0.122$ & $0.092$ & $0.184$ &  & $-0.088$ & $0.092$ & $0.338$ &  & $-0.085$ & $0.092$ & $0.356$ &  & $-0.096$ & $0.092$ & $0.298$ & \\
\noalign{\smallskip}
\midrule
\noalign{\smallskip}
\multirow{4}{*}{\texttt{party\_Republican:region}} & \texttt{Midwest} & $0.219$ & $0.059$ & $<0.001$ & \multirow{4}{*}{$<0.001$} & $0.215$ & $0.059$ & $<0.001$ & \multirow{4}{*}{$<0.001$} & $0.260$ & $0.059$ & $<0.001$ & \multirow{4}{*}{$<0.001$} & $0.246$ & $0.059$ & $<0.001$ & \multirow{4}{*}{$<0.001$} & $0.251$ & $0.059$ & $<0.001$ & \multirow{4}{*}{$<0.001$} & $0.240$ & $0.059$ & $<0.001$ & \multirow{4}{*}{$<0.001$}\\
 & \texttt{Southeast} & $0.073$ & $0.066$ & $0.270$ &  & $0.067$ & $0.066$ & $0.308$ &  & $0.078$ & $0.066$ & $0.239$ &  & $0.081$ & $0.066$ & $0.220$ &  & $0.065$ & $0.066$ & $0.326$ &  & $0.060$ & $0.066$ & $0.361$ & \\
 & \texttt{South} & $0.136$ & $0.074$ & $0.067$ &  & $0.141$ & $0.074$ & $0.057$ &  & $0.131$ & $0.074$ & $0.077$ &  & $0.100$ & $0.074$ & $0.179$ &  & $0.148$ & $0.074$ & $0.046$ &  & $0.104$ & $0.074$ & $0.159$ & \\
 & \texttt{West} & $0.193$ & $0.066$ & $0.004$ &  & $0.182$ & $0.066$ & $0.006$ &  & $0.193$ & $0.066$ & $0.004$ &  & $0.220$ & $0.066$ & $<0.001$ &  & $0.207$ & $0.066$ & $0.002$ &  & $0.224$ & $0.066$ & $<0.001$ & \\
\noalign{\smallskip}
\midrule
\noalign{\smallskip}
\multirow{2}{*}{\texttt{party\_Republican:generation}} & \texttt{Boomer} & $-0.138$ & $0.065$ & $0.034$ & \multirow{2}{*}{$0.011$} & $-0.177$ & $0.065$ & $0.007$ & \multirow{2}{*}{$0.007$} & $-0.139$ & $0.065$ & $0.032$ & \multirow{2}{*}{$0.043$} & $-0.156$ & $0.065$ & $0.017$ & \multirow{2}{*}{$0.018$} & $-0.136$ & $0.065$ & $0.036$ & \multirow{2}{*}{$0.036$} & $-0.138$ & $0.065$ & $0.035$ & \multirow{2}{*}{$0.006$}\\
 & \texttt{Gen X} & $-0.193$ & $0.084$ & $0.022$ &  & $-0.156$ & $0.084$ & $0.065$ &  & $-0.126$ & $0.084$ & $0.137$ &  & $-0.145$ & $0.084$ & $0.086$ &  & $-0.141$ & $0.084$ & $0.095$ &  & $-0.217$ & $0.084$ & $0.010$ & \\
\noalign{\smallskip}
\midrule
\noalign{\smallskip}
\multirow{2}{*}{\texttt{party\_Republican:experience}} & \texttt{(1,10]} & $0.015$ & $0.078$ & $0.845$ & \multirow{2}{*}{$0.925$} & $-0.056$ & $0.078$ & $0.476$ & \multirow{2}{*}{$0.406$} & $0.012$ & $0.078$ & $0.879$ & \multirow{2}{*}{$0.960$} & $-0.004$ & $0.078$ & $0.960$ & \multirow{2}{*}{$0.984$} & $0.044$ & $0.078$ & $0.572$ & \multirow{2}{*}{$0.819$} & $0.018$ & $0.078$ & $0.814$ & \multirow{2}{*}{$0.969$}\\
 & \texttt{(0,1]} & $-0.044$ & $0.137$ & $0.747$ &  & $-0.165$ & $0.137$ & $0.230$ &  & $0.035$ & $0.137$ & $0.798$ &  & $-0.024$ & $0.137$ & $0.861$ &  & $-0.031$ & $0.137$ & $0.820$ &  & $0.015$ & $0.137$ & $0.914$ & \\
\noalign{\smallskip}
\midrule
\noalign{\smallskip}
\multirow{7}{*}{\texttt{party\_Republican:religion}} & \texttt{Catholic} & $-0.006$ & $0.070$ & $0.926$ & \multirow{7}{*}{$0.054$} & $-0.024$ & $0.070$ & $0.726$ & \multirow{7}{*}{$0.081$} & $-0.009$ & $0.070$ & $0.898$ & \multirow{7}{*}{$0.006$} & $-0.011$ & $0.070$ & $0.874$ & \multirow{7}{*}{$0.019$} & $0.001$ & $0.070$ & $0.988$ & \multirow{7}{*}{$0.048$} & $-0.014$ & $0.070$ & $0.840$ & \multirow{7}{*}{$0.012$}\\
 & \texttt{Presbyterian} & $-0.212$ & $0.082$ & $0.009$ &  & $-0.169$ & $0.082$ & $0.039$ &  & $-0.242$ & $0.082$ & $0.003$ &  & $-0.234$ & $0.082$ & $0.004$ &  & $-0.202$ & $0.082$ & $0.013$ &  & $-0.242$ & $0.082$ & $0.003$ & \\
 & \texttt{Baptist} & $-0.003$ & $0.102$ & $0.976$ &  & $-0.059$ & $0.102$ & $0.564$ &  & $-0.052$ & $0.102$ & $0.611$ &  & $0.027$ & $0.102$ & $0.788$ &  & $0.021$ & $0.102$ & $0.836$ &  & $-0.065$ & $0.102$ & $0.522$ & \\
 & \texttt{Jewish} & $0.000$ & $3.211$ & $1.000$ &  & $0.000$ & $3.211$ & $1.000$ &  & $0.000$ & $3.211$ & $1.000$ &  & $0.000$ & $3.211$ & $1.000$ &  & $0.000$ & $3.211$ & $1.000$ &  & $0.000$ & $3.211$ & $1.000$ & \\
 & \texttt{Methodist} & $0.144$ & $0.090$ & $0.111$ &  & $0.208$ & $0.090$ & $0.021$ &  & $0.204$ & $0.090$ & $0.024$ &  & $0.194$ & $0.090$ & $0.032$ &  & $0.191$ & $0.090$ & $0.034$ &  & $0.175$ & $0.090$ & $0.053$ & \\
 & \texttt{Lutheran} & $0.188$ & $0.104$ & $0.072$ &  & $0.158$ & $0.104$ & $0.130$ &  & $0.229$ & $0.104$ & $0.028$ &  & $0.188$ & $0.104$ & $0.071$ &  & $0.183$ & $0.104$ & $0.079$ &  & $0.206$ & $0.104$ & $0.049$ & \\
 & \texttt{Mormon} & $0.135$ & $0.135$ & $0.317$ &  & $0.036$ & $0.135$ & $0.791$ &  & $0.105$ & $0.135$ & $0.435$ &  & $0.061$ & $0.135$ & $0.651$ &  & $0.046$ & $0.135$ & $0.732$ &  & $0.104$ & $0.135$ & $0.439$ & \\
\noalign{\smallskip}
\bottomrule
\noalign{\medskip}
\end{tabular}